%% file: main.tex
\documentclass[12pt]{amsart}

\input{preamble}

\title{\textbf{Compositional quantum heuristics for max-clique detection}}

\author{Tiffany Duneau}
\address[T.\ Duneau]{Quantinuum\\Oxford\\17 Beaumont St\\ OX1 2NA\\Oxford\\United Kingdom}
\email{\href{mailto:tiffany.duneau@quantinuum.com}{tiffany.duneau@quantinuum.com}}
\author{Colin Krawchuk}
\address[C.\ Krawchuk]{Quantinuum\\Oxford\\17 Beaumont St\\ OX1 2NA\\Oxford\\United Kingdom}
\email{\href{mailto:colin.krawchuk@quantinuum.com}{colin.krawchuk@quantinuum.com}}
\author{Anna Pearson}
\address[A.\ Pearson]{Quantinuum\\Oxford\\17 Beaumont St\\ OX1 2NA\\Oxford\\United Kingdom}
\email{\href{mailto:anna.pearson@quantinuum.com}{anna.pearson@quantinuum.com}}

\begin{document}

\maketitle

\input{Paper/abstract/abstract}

\input{Paper/intro/intro}

\input{Paper/background/background}
\input{Paper/theory/loss_invariance}
\input{Paper/theory/qgnns}

\input{Paper/experiments/tasks}

\input{Paper/theory/pine}

\input{Paper/experiments/maxcliquemodels}
\input{Paper/experiments/findkmodels}
\input{Paper/experiments/conclusion}

\clearpage

\bibliographystyle{alpha}
\bibliography{refs.bib}

\clearpage

\appendix
\input{Paper/appendix/A}

\input{Paper/appendix/B}

\input{Paper/appendix/Extra_plots}

\end{document}

%% file: preamble.tex
\usepackage[utf8]{inputenc}
\usepackage[T1]{fontenc}
\usepackage{lmodern}
\usepackage[canadian]{babel}
\usepackage[babel=true]{microtype}

\usepackage{amsmath, amssymb, amsfonts, amscd, mathtools, commath}
\usepackage{newtxmath}
\usepackage{amstext, array}
\usepackage{dsfont}
\usepackage{ytableau}

\numberwithin{equation}{section}

\usepackage{graphicx, float, standalone, pict2e, epic, epstopdf, subcaption}
\usepackage{tikz-cd}

\usetikzlibrary{
quantikz2,
arrows.meta,
decorations.pathmorphing,
backgrounds,
positioning,
fit,
trees,
shapes
}

\usepackage[pdflang=en-UK,
colorlinks,
urlcolor=midgreen,
linkcolor=midgreen,
citecolor=midgreen]{hyperref}

\usepackage{cleveref}

\usepackage{algorithm}
\usepackage{algpseudocode}

\usepackage{titlesec}

\newcommand{\sectionrule}{
\vspace{0.5em}
\noindent\tikz[baseline]{\draw[thick, color=gray!60] (0,0) -- (16,0);}
\par
}

\titleformat{\section}
{\normalfont\fontsize{16pt}{16pt}\bfseries\color{darkgreen}}
{}
{0pt}
{}[\vspace{-0.5em}\sectionrule]

\titleformat{\subsection}
{\normalfont\large\bfseries\color{darkpurple}}
{}
{0pt}
{}

\titleformat{\subsubsection}
{\normalfont\bfseries\color{darkpurple}}
{}
{0pt}
{}

\usepackage{fancyhdr}
\usepackage[a4paper,margin=1in,headheight=14pt]{geometry}

\fancyhfoffset[L]{0pt}
\fancyhfoffset[R]{0pt}

\usepackage{commands}


%% file: Paper/abstract/abstract.tex
\section{Abstract}

Quantum machine learning holds the promise of combining the success of classical machine learning methods with the power of quantum computing, however one of the largest obstacles facing the field is the problem of barren plateaus. Parameterised quantum circuits offer a flexible framework for developing quantum machine learning models, but their practicality is constrained by a trade-off between trainability and classical simulability. In general, circuits that are sufficiently expressive to model complex behaviour often exhibit barren plateaus, where gradients vanish and optimisation fails. 

In this work we investigate a compositional approach to mitigate this trade-off by assembling larger quantum models from smaller subcomponents. To ensure trainability of these subcomponents, we describe a framework for constructing group-invariant loss functions, which introduce symmetry-induced inductive bias and lead to improved gradient behaviour and generalisation. In particular, we use this framework to design permutation-equivariant quantum graph neural networks for identifying maximal cliques in graphs. The models we construct exhibit superior training gradients through symmetry-induced bias, and our experiments demonstrate that the trained models generalise to larger, more complex problem instances. 

Finally, inspired by Quantum-Informed Recursive Optimisation Algorithms \cite{QIROA}, we implement a recursive hybrid quantum-classical heuristic using the learned quantum models to guide a classical search procedure, demonstrating improved inference accuracy and scalability. Together, these results suggest that compositional circuits could be a viable pathway towards scalable quantum learning models that remain challenging to reproduce classically.

%% file: Paper/intro/intro.tex
\section{Introduction}

\begin{figure}[h]
    \centering
    \begin{subfigure}[b]{0.2\linewidth}
        \centering
        \includegraphics[width=\linewidth]{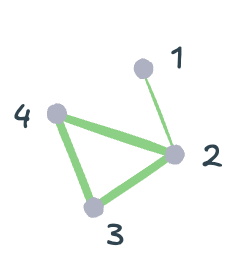}
        \caption{}
        \label{fig:experiment_schema/g}
    \end{subfigure}
    \hspace{3em}
    \begin{subfigure}[b]{0.2\linewidth}
        \centering
        \includegraphics[width=\linewidth]{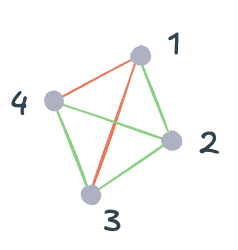}
        \caption{}
        \label{fig:experiment_schema/g_c}
    \end{subfigure}
    \hfill
    \begin{subfigure}[b]{0.8\linewidth}
        \centering
        \includegraphics[width=\linewidth]{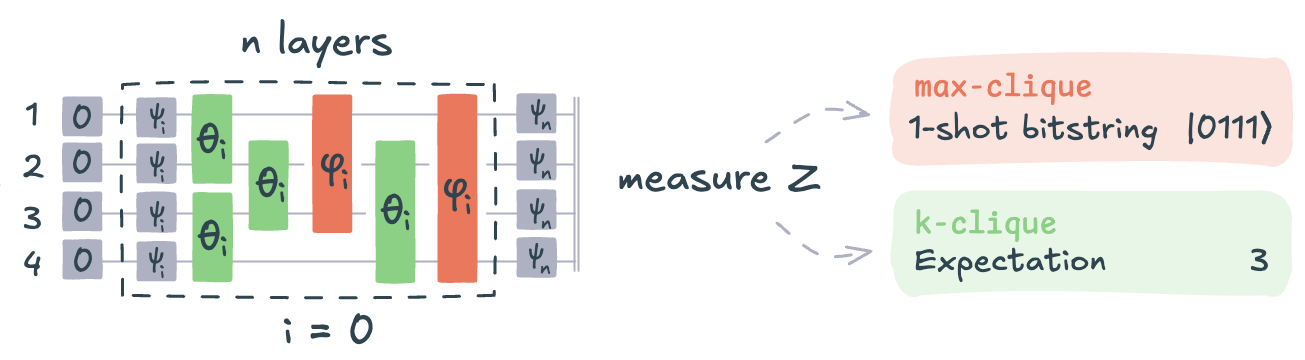}
        \caption{}
        \label{fig:experiment_schema/ansatz}
    \end{subfigure}
    \caption{The experiment setup at a glance: (A) A graph, with the maximum sized 3-clique highlighted. (B) Adding complement edges (orange). (C) The graph encoded as a quantum circuit. Qubits are labelled by the nodes. For each edge, we add a gate parametrised by its colour. We interpret the output as either a bitrstring identifying a maximum clique, or the size of the max clique.}
    \label{fig:experiment_schema}
\end{figure}

Parameterised quantum circuits (PQCs) are a promising candidate for near-term quantum applications due to their resilience to noise and suitability for diverse tasks, such as variational quantum eigensolvers \cite{peruzzo2014variational}, variational quantum time evolution \cite{yuan2019theory}, quantum approximate optimization algorithm \cite{QAOA}, among others. Additionally, their ability to approximate complex mappings between data and measurement outcomes allows PQCs to act as natural analogues of classical neural networks.

In the context of machine learning, demonstrating an advantage with PQCs commonly requires identifying a task where the circuit can be efficiently trained and evaluated on a quantum device without being classically simulable. This can result in a computational advantage if the labelling function is intractable for classical computers, or a learning theoretic advantage if a classical computer is unable to efficiently identify the labelling function \cite{dunjko}. In either case, the goal is to isolate a setting where quantum circuits can be trained effectively while maintaining a separation in computational or statistical efficiency compared to classical algorithms.

The main barrier to using PQCs is that they are known to exhibit barren plateaus, where the magnitude of loss-function gradients decreases exponentially with the size of the system \cite{McClean2018BP}. As a result, gradient-based optimisation becomes infeasible for large circuits. Various strategies have been proposed to construct PQCs without barren plateaus by introducing design constraints such as locality, symmetry, or restricted gate sets such as matchgates \cite{holmes2022connecting,larocca2023theory,marrero2021entanglement}. However, while these modifications can improve trainability, in many cases the addition of this extra structure has enabled classical simulation \cite{jozsa2008matchgates, takahashi2020classically, goh2025lie}.

One proposal for mitigating this tradeoff is to construct models that can be trained classically, and only require a quantum computer at test time. 
For example, by building models from instantaneous quantum polynomial (IQP) circuits, expectations can be computed in polynomial time and used to train the models, while sampling to extract solutions is believed to be hard \cite{recio-armengol_train_2025, bremner_classical_2011}. However, in many cases, practical implementations on NISQ devices may become simulable after all due to compilation overheads \cite{placidi_impact_2026}.

Here, we consider a second approach via \textit{compositional} models, where the PQC is first trained on smaller subcircuits. These local circuits are then composed to provide meaningful inferences at large scales \cite{duneau_scalable_2024}. For this strategy to be effective, the constituent subcircuits must remain trainable and therefore we develop a framework for constructing group-invariant loss functions and show experimental evidence of improved training dynamics.

The obstacle to the approach outlined above is finding tasks where evaluation can be understood through smaller, barren plateau-free contexts, where the problem has the requisite symmetry for trainability, and where the desired function is classically challenging to compute. In this work we design quantum graph neural network models for detecting maximally sized cliques in simple graphs (as per the setup in \autoref{fig:experiment_schema}). This setting allows us to take advantage of the recursive nature of clique detection, as well as existing theoretical results on the trainability of permutation equivariant quantum circuits.

Furthermore, we demonstrate that our models are scalable - allowing meaningful predictions to be made on problem instances that are more challenging than those on which the model was trained. In particular, we show that symmetrised models performs better on the generalisation task than the unsymmetrised model, thus demonstrating compositional generalisation with low training overhead.

%% file: Paper/background/background.tex
\section{Setup}

Let $\mathcal{X}$ be a data domain, $\Theta$ a parameter space, and $\mathcal{H}$ a finite-dimensional Hilbert space.

\begin{definition}\label{def:pqcs}
    We consider layered parameterised quantum circuits of the form
    \[
    \pqc = \Pi_{\ell = 1}^{L}\Pi_{k=1}^K e^{ \theta_{\ell,k}H_k}
    \]
    where $-i\textbf{H}_k$ is the gate hamiltonian and $\theta_{\ell,k}$ are tunable parameters.
\end{definition}

In this work we focus on QML models formed from PQCs.

\begin{definition} A \emph{quantum neural network} model $\mathcal{M}$ is specified by
\[
\mathcal{M}=\left(\rho(X), \bf{U}(X, \boldsymbol{\theta}),  \obs(X), \Phi\right)\]
 
where for each entry $X \in \mathcal{X}$:
\begin{itemize}
  \item $\rho(X) \in i\mathfrak{u}(\mathcal{H})$ is a density operator (positive semi-definite and trace one),
  \item $\mathbf{U}(X,\theta) \in \mathbf{U}(\mathcal{H})$ is a unitary circuit, possibly parameterised by $\theta \in \Theta$,
  \item $\obs(X) \in i\mathfrak{u}(\mathcal{H})$ is an observable,
  \item $\Phi$ is a fixed classical post-processing functional.
\end{itemize}

The model output is the function
\[
\Phi(X, \boldsymbol{\boldsymbol{\theta}} ):=\Phi\left(\rho(X), \bf{U}(X, \boldsymbol{\theta}),  \obs(X)\right)  \in \mathcal{Y}
\]
where the output space $\mathcal{Y}$ is determined by the choice of $\Phi$. 
\end{definition}

Note that each component of the model may depend on $X$. Whenever a component is independent of the data entry (say $\rho(X) = \rho, \text{ for all } X \in \mathcal{X})$ we suppress $X$ in the notation. A typical loss function considers the expectation value of the observable $\obs$ with respect to the parametrised circuit:
\[\mathcal{L}(X,\boldsymbol{\theta})
=
\trace{
\obs(X)\,\mathbf{U}(X,\boldsymbol{\theta})\,
\rho(X)\,
\mathbf{U}(X,\boldsymbol{\theta})^{\dagger}
}.\]
Typically, the loss function $\mathcal{L}$ is distinct from the model output $\Phi$. Indeed we explicitly do not restrict the loss to operate only over the label space $\mathcal{Y}$ in order to account for QAOA-like setups in which $\mathcal{L}$ considers an expectation value while $\Phi$ represents sampling from the quantum state defined by the model components.

We assess the trainability of the parameterised models compared in this paper by comparing the magnitude of their loss-function gradients. Denote by $\partial_{\mu}\mathcal{L}$ the partial derivative with respect
to a circuit parameter $\mu\in\boldsymbol{\theta}$.

\begin{definition}[Barren plateau]
A parameterised quantum model is said to exhibit a \emph{barren plateau} (with respect to a specified distribution over parameters, typically random initialisation) if
\[\mathbb{E}\!\left[\partial_{\mu}\mathcal{L}\right] = 0,\]
and the variance
\[\mathrm{Var}\!\left(\partial_{\mu}\mathcal{L}\right)\]
decays exponentially with the system size (number of qubits). In this regime, estimating gradients to constant relative precision requires an exponential number of circuit evaluations.
\end{definition}

This definition follows the standard operational notion introduced by McClean et al. \cite{MBS+18}, and captures the practical obstruction to
gradient-based optimisation on near-term quantum hardware. Since gradients are estimated from finite measurement statistics, small variance directly translates into prohibitive shot complexity.

Throughout this work, \textit{trainability} refers specifically to the absence of exponential gradient suppression arising from expressibility-induced concentration, rather than from shot noise alone or optimisation pathologies such as local minima.

Previous works have shown that barren plateaus are closely tied to the size of the \emph{dynamical lie algebra} (DLA) associated with a parameterised quantum circuit \cite{Fontana2023, holmes2022connecting} and that symmetric ansatz effectively constrain the DLA dimension \cite{Schatzki2022}. Motivated by this, we outline a general framework for constructing PQC models with symmetric loss functions.

%% file: Paper/theory/loss_invariance.tex
\section{A Framework for Group Invariant Loss Functions} 

To make the compositional PQC approach viable, we require subcircuits that admit stable and informative training dynamics. We therefore seek structural constraints that improve optimisation while preserving expressivity.

This section formalises how symmetry constraints can be used to design QNN models with loss functions that are invariant with respect to a group action. This condition has the effect of restricting the dynamical lie algebra of the model. In \autoref{sec:conclusion} we provide empirical evidence that symmetry invariant loss functions exhibit better training dynamics while enabling compositional scaling. The relevant background on group representations can be found in \autoref{app:background}.

\subsection{Semi-Symmetric Circuits}\label{sec:semi_sym_circs}

Suppose we have a unitary representation of a group $\phi: G \to \mathbf{U}(\mathcal{H})$. Composing with the adjoint representation of $ \mathbf{U}(\mathcal{H})$ defines a $G$-action on $i\mathfrak{u}(\mathcal{H})$:
\[
\mathrm{Ad} \circ \phi(g):  \ H \mapsto \phi(g) H \phi(g)^\dagger \quad \text{for all } g \in G.
\]
Now assume we have a group action of $G$ on the data domain
\[
G \times \mathcal{X} \to \mathcal{X}; \  (g, X) \mapsto g \cdot X .\]

Viewing the density operator as a map
\[
\rho(\placeholder ): \mathcal{X}  \to i\mathfrak{u}(\mathcal{H}); \ X \mapsto \rho(X)
\]
we say that $\rho$ is {$G$-equivariant} if 
\[
\rho(g \cdot X)=\phi(g) \rho(X) \phi(g)^{\dagger}
\]
for all  $g \in G$. This is captured by the commutative diagram below:
\[
\begin{tikzcd}[
    row sep=2.5em,
    column sep=3.5em,
    nodes={inner sep=1.0ex},  
    arrows={shorten <= 0.3em, shorten >= 0.3em}
]
\mathcal{X} \arrow[r,"\rho(\placeholder )"] \arrow[d,"g\cdot "'] & i\mathfrak{u}(\mathcal{H}) \arrow[d,"\mathrm{Ad} \circ \phi(g)"] \\
\mathcal{X} \arrow[r,"\rho(\placeholder )"'] & i\mathfrak{u}(\mathcal{H})
\end{tikzcd} 
\]
Similarly, we call an observable $G$-equivariant when
\[
\obs(g \cdot X)=\phi(g)\obs(X) 
\phi(g)^{\dagger}
\]
for all $g \in G$. In the case of circuits, we consider equivariance with respect to the conjugation action of $\phi(g)$ on $\mathbf{U}(\mathcal{H})$:
\[
U(g \cdot X, \theta)=\phi(g) U(X, \theta) \phi(g)^{\dagger}\]

Analogously, we say that $\rho, \pqc, \  \obs$ are $G$-invariant if they are invariant under both the data and conjugation actions: 
\begin{align*}
\rho(g \cdot X) = \rho(X) \quad &\text{and} \quad \phi(g)\rho(X)\phi(g)^\dagger = \rho(X), \\
U(g \cdot X, \theta) = U(X, \theta) \quad &\text{and} \quad  \phi(g)U(g \cdot X, \theta)\phi(g)^\dagger = U(X, \theta) \\
\obs(g \cdot X) = \obs(X) \quad &\text{and} \quad \phi(g)\obs( X) \phi(g)^\dagger= \obs(X)
\end{align*}
respectively. 


\begin{definition} A QNN is \textit{semi-symmetric} if some subset of $\{\mathbf{U}, \obs, \rho\}$ is $G$-equivariant, and all other components are $G$-invariant.
\end{definition}

The following result provides a general framework for constructing $G$-equivariant QNN models. 

\begin{prop}\label{prop:semi_sym_invariance}
   Suppose $\mathcal{M}$ is a semi-symmetric QNN model with expectation loss,
   \[
\mathcal{L}(X,\boldsymbol{\theta})=\trace{\obs(X) \bf{U}(X, \theta) \rho(X) \bf{U}(X, \theta)^{\dagger}}
\]
then the model $\mathcal{M}$ is $G$-invariant.
\end{prop}

\begin{figure}[h]
    \centering
    \includegraphics[width=0.3\linewidth]{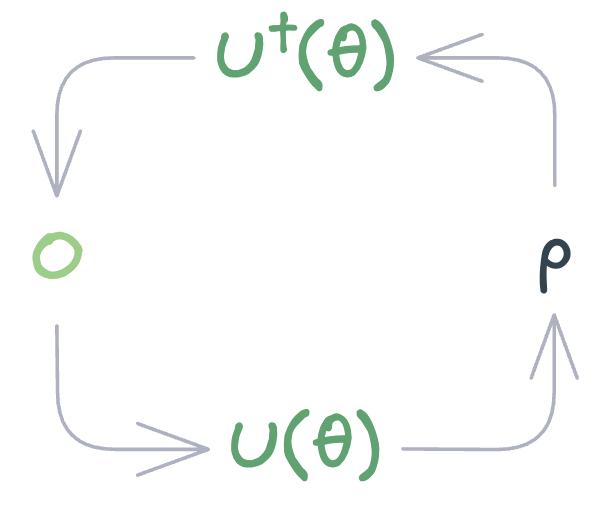}
    \hspace{3em}
    \includegraphics[width=0.45\linewidth]{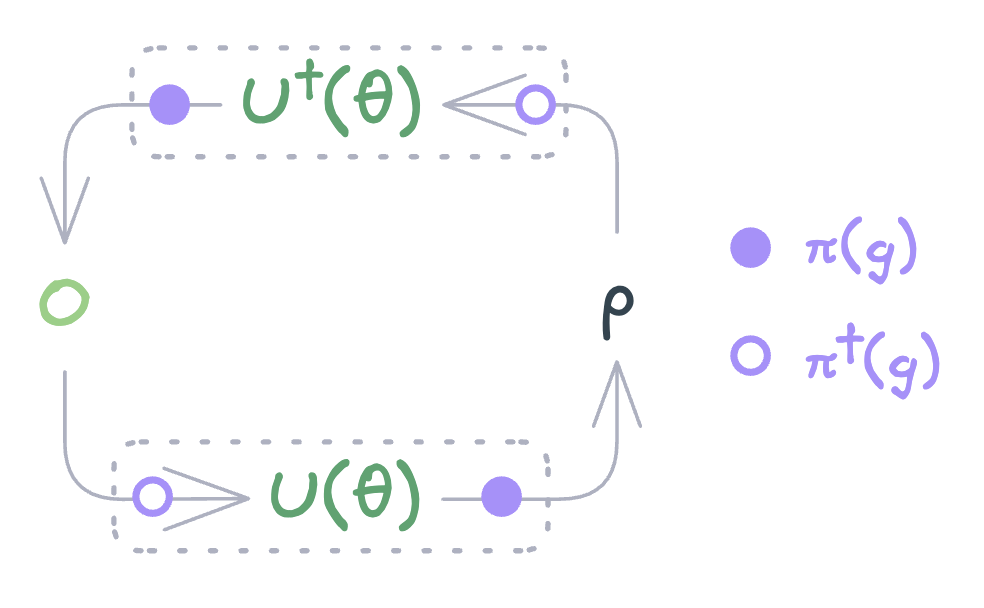}
    \caption{The cyclicity of the trace. The group actions are either absorbed or pass through the equivariant components of the circuit, resulting in an invariant loss.}
    \label{fig:trace}
\end{figure}

In order to derive trainability guarantees for PQCs (in the spirit of \cite[Theorem 2.8]{Fontana2023}) we would require fully symmetric QNN models where each component is $G$-invariant. In practice, this constraint is unrealistic as it leads to model constructions without awareness of individual data entries. However, as demonstrated by our results in Section \autoref{subsec:training_dynamics}, the weaker Semi-Symmetric constraint provides improved training dynamics while preserving meaningful dependence on the input problem. 

%% file: Paper/theory/qgnns.tex
\section{Quantum Graph Neural Networks}\label{sec:qgnns}

The previous section introduced a framework for constructing invariant loss functions with respect to arbitrary group actions. As an illustration, we now specialise to the case of the permutation action of the symmetric group $S_n$ on a quantum circuit. This section describes a permutation Semi-Symmetric model that will be used in \autoref{sec:tasks} to create quantum heuristics for graph optimisation problems.

Let $\Gamma = (\Gamma_v,\Gamma_e)$ be a simple graph on $n$-vertices with vertex set $\Gamma_v$ and edge set $\Gamma_e$. The ambient data domain of labelled graphs on $n$ vertices carries a natural action of the symmetric group $S_n$, given by relabelling vertices. Explicitly, for a permutation $\sigma \in S_n$, the relabelled graph $\sigma \cdot \Gamma$ is obtained by mapping each vertex $i \in \Gamma_v$ to $\sigma(i)$ and correspondingly permuting edges. For a fixed graph $\Gamma$, the subgroup of graph automorphisms
\[
\mathrm{Aut}(\Gamma) \;\leq\; S_n
\]
consists of permutations that preserve adjacency, which is the maximal symmetry group intrinsic to that specific graph.

\subsection{Graph-Based Circuit Ansatz}

We now introduce the graph-based variational circuit structures used in the later sections. Let $\mathcal{H} = (\mathbb{C}^2)^{\otimes n}$, with one qubit associated to each vertex of $\Gamma$. In this work we only use a single qubit per node, but in general these encodings can be applied to multiple qubits (where  $\mathcal{H} = (\mathbb{C}^d)^{\otimes n}$).

\begin{definition}[Graph-based ansatz]
The ansatze we consider here can all be framed in terms of the following structure:

\begin{equation}\label{eq:graph_ansatz}
    \mathbf{U}(\boldsymbol{\theta}) =
    \prod_{\ell=1}^{L}
    \left[
        \exp\!\left(\sum_{e\in \Gamma_e} H_{e}(\vec{\beta_\ell}) \right)\
        \exp\!\left(\sum_{\bar{e} \notin \Gamma_e} H_{\bar{e}}(\vec{\gamma_\ell}) \right)
        \exp\!\left(\sum_{v \in \Gamma_v} H_{v}(\vec{\alpha_\ell}) \right)\
    \right]
\end{equation}

Where we have expanded $\pqc$ as repeated layers of a graph embedding. The three components of each layer are understood to embed node information via the node hamiltonian $H_v$, and edges via the edge and anti-edge hamiltonians $H_e$ and $H_{\bar{e}}$. Each of these hamiltonians are parametrised independently at each layer.    
\end{definition}

Below, we consider a selection of edge hamiltonians that give rise to different equivariance and invariance properties of the resulting circuit $\pqc$. Note that we adopt the convention of writing $X_v, Y_v, Z_v$ to denote the Pauli operators acting on qubit $v$, and expand vector components in the superscript: $\vec{\beta}_{\ell} := [\beta^{(0)}_{\ell}, ..., \beta^{(n)}_{\ell}]$.

\begin{definition}[\Rook]\label{def:rook}
    Symmetric (K)complete graph Ansatz.
    We consider layered parametrised quantum circuits of the form in \cref{eq:graph_ansatz}, where:

\begin{align*}
    H_{v}(\vec{\alpha_{\ell}}) =  \alpha^{(0)}_{\ell}\, X_v + \alpha^{(1)}_{\ell}\, Z_v + \alpha^{(2)}_{\ell}\, X_v
    &&
    H_{e}(\vec{\beta}_{\ell}) = \beta^{(0)}_{\ell}\, Z_i Z_j + \beta^{(1)}_{\ell} (Z_i + Z_j)
    \\
    &&
    H_{\bar{e}}(\vec{\gamma}_{\ell}) = \gamma^{(0)}_{\ell}\, Z_i Z_j + \gamma^{(1)}_{\ell} (Z_i + Z_j)
\end{align*}
    The node hamiltonian $H_{v}$ acts identically on each qubit with layer-specific parameters $\alpha^{(X)}_{\ell}$, $\alpha^{(Y)}_{\ell}$, and $\alpha^{(Z)}_{\ell}$. 
    The \emph{edge layer} consists of $ZZ$ phase gadgets $\exp(\beta^{(0)}_{\ell} Z_i Z_j)$ applied to each edge $e = (i,j)\in\Gamma_e$, alongside single-body components $\exp(\beta^{(1)}_{\ell} Z_i)$, $\exp(\beta^{(1)}_{\ell} Z_j)$ for each node involved in the edge.
    Similarly, the \emph{anti-edge layer} consists of $ZZ$ phase gadgets
    $\exp(\gamma^{(0)}_{\ell} Z_i Z_j)$ and $\exp(\gamma^{(1)}_{\ell} Z_i)$, $\exp(\gamma^{(1)}_{\ell} Z_j)$ applied to all non-edges $\bar{e} =(i,j)\notin\Gamma_e$.
    The complete set of trainable parameters shared across vertices, edges, and non-edges are 
    \[
    \boldsymbol{\theta}
    =
    \{
    \vec{\alpha}_\ell,
    \vec{\beta}_{\ell},
    \vec{\gamma}_{\ell},
    \}_{\ell=1}^L
    \] 
    respectively.
\end{definition}

This structure matches Definition \ref{def:pqcs}, with gate hamiltonians determined by the graph topology. Importantly, parameters are {tied across orbits} of the vertex and edge sets, rather than assigned independently.

The expressivity of this ansatz depends on the particular graph that is being encoded. We add the anti-edge layer in order to increase the base expressivity, such that the least expressive graph ansatz remains capable of encoding optimal solutions to typical graph problems. In particular, this intends to mitigate known limitations of QAOA-like graph ansatze due to locality \cite{farhi_quantum_2020}. By encoding the graph within a complete graph, each node can be correlated to any other node in the graph independently of the choice of layers, at a polynomial cost to circuit depth.

\begin{remark}
    Consider a graph with empty edge set. Without anti-edges, the resulting circuit would decompose into a product state over $n$ qubits:
    \[
    \mathbf{U}(\boldsymbol{\theta})
    = 
    \prod_{\ell = 1}^L \left[
        \exp\!\left(\sum_{v \in \Gamma_v} H_{v}(\vec{\alpha_\ell}) \right)\
    \right]
    = 
    \prod_{\ell = 1}^L \left[
    \exp\!\left(H_{v}(\vec{\alpha_\ell}) \right)^{|\Gamma_v|}
    \right]
    \]
    More generally, the ansatz without anti-edges decomposes as a product state over the disconnected components of the graph.
    For a task whose optimal solution does not directly decompose into a product of the optimal solutions over the the disconnected subgraphs, this ansatz cannot capture the optimal solution. The max-clique problem we consider in \autoref{sec:tasks} is an example of such a task, while graph colouring or maximum cut would admit the simpler ansatz.
\end{remark}

\begin{prop}[Graph-automorphism equivariance]\label{prop:rook_is_symmetric} For each data point $\Gamma$, the \Rook
ansatz is $S_n$-equivariant.
\end{prop}

Using results \autoref{prop:rook_is_symmetric} and \autoref{prop:semi_sym_invariance} we can build an $S_n$-invariant model by combining the \Rook ansatz with an $S_n$-invariant initial state and observable. Next, we introduce a more general graph ansatz that will be used as an unsymmetrised comparison to \Rook-based models. 

\begin{definition}[\MF]
An alternative `asymmetric' layered parametrised complete graph-based ansatz taking the form of \autoref{eq:graph_ansatz} is given by:

\begin{align*}
    H_{v}(\vec{\alpha}_\ell) = \alpha^{(0)}_{\ell}\, X_v + \alpha^{(1)}_{\ell}\, Z_v + \alpha^{(2)}_{\ell}\, X_v
    &&
    H_{e}(\vec{\beta}_\ell) = H^{(M)}_{ij}(\vec{\beta}_{\ell})
    \\&&
    H_{\bar{e}}(\vec{\gamma}_\ell) = H^{(M)}_{ij}(\vec{\gamma}_{\ell})
\end{align*}

The node hamiltonian $H_{v}$ acts identically on each qubit with layer-specific parameters $\alpha^{(X)}_{\ell}$, $\alpha^{(Y)}_{\ell}$, and $\alpha^{(Z)}_{\ell}$. The \emph{edge layer} is a generic two-qubit unitary acting on qubits $e = (i,j)\in\Gamma_e$, approximated by $M$ layers of the SIM4 ansatz  (Circuit 4 of \cite{sim_expressibility_2019}), with shared parameter vector $\vec{\beta}_{\ell}$. The \emph{anti-edge layer} is an independent, generic two-qubit unitary acting on qubits $\bar{e} = (i,j)\notin\Gamma_e$, also approximated by $M$ layers of the SIM4 ansatz, with shared parameter vector $\vec{\gamma}_{\ell}$.

The trainable parameters shared across vertices, edges, and anti-edges are
\[
\boldsymbol{\theta}
=
\{
\vec{\alpha}_\ell, \vec{\beta}_{\ell}, \vec{\gamma}_{\ell}
\}_{\ell=1}^L
\]
respectively.    

\end{definition}

\begin{figure}
    \centering
    \begin{subfigure}{0.4\linewidth}
        \centering
        \includegraphics[scale=0.75]{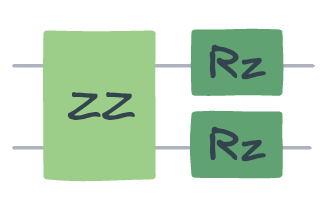}
        \caption{}
        \label{fig:rook}
    \end{subfigure}
    \hspace{2em}
    \begin{subfigure}{0.4\linewidth}
        \centering
        \includegraphics[scale=0.75]{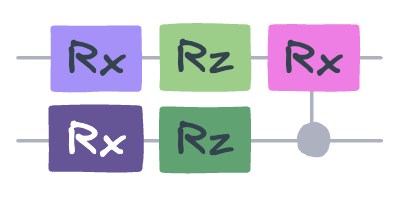}
        \caption{}
        \label{fig:sim4}
    \end{subfigure}
    \caption{A representation of the edge ansatze $H_e$, $H_{\bar{e}}$ used in this work. Gates with shared parameters are depicted in the same colour.
    (A) Equivariant ansatz. Each layer has 2 trainable parameters.
    (B) The SIM4 circuit ansatz of \cite{sim_expressibility_2019}, applied to two qubits. Each layer has 5 trainable parameters.}
\end{figure}

\begin{remark}\label{rem:MF_not_equivariant}
    The only guaranteed symmetry group of the \MF ansatz is the trivial group. This is because generic $2$-qubit operators need not commute, and so any reordering of edges may result in a distinct circuit. 
\end{remark}

The \MF is intentionally included as a control in order to illustrate a lack of symmetry in comparison to \Rook.

\subsection{Initial States}

When using input dependent circuits like \Rook or \MF, the initial state is taken to be a product state such as $\ket{0}^{\otimes n}$ or $\ket{+}^{\otimes n}$, as this provides an $S_n$-invariant initial state.
Here we consider trainable product states $\ket{\psi}^{\otimes n}$ that are implemented via an Euler decomposition $\ket{\psi} = \exp(\alpha_0^{(0)}X + \alpha_0^{(1)}Z + \alpha_0^{(2)}X)$.

While not considered in this work, one can define circuits that are $S_n$-invariant, although completely independent of the data entry (see \cite{Schatzki2022} for example). In this case, graph states are used to encode the problem information. These graph states are $S_n$-equivariant, and so such models are also Semi-Symmetric. 

\subsection{Observables}
The elements we have considered so far, $\rho$ and $\pqc$, act as an embedding of the graph. The observable we choose specifies how we extract information to solve the task. The choice of observable also influences the possible loss functions. Our main constraint is to maintain an $S_n$-invariant observable that can be used in conjunction with equivariant circuits. 

\begin{definition}[\Bitstring]\label{def:bitstring}
    The \Bitstring encoding simply interprets each computational basis state as an indicator function on the corresponding vertices. Each bitstring thus identifies a specific subset of $\Gamma_v$, which we interpret as a candidate solution. The resulting observable is then the identity:
    $$\obs = I$$
\end{definition}

\begin{definition}[\Mountain] $k$-valued hamming weight.
    This observable partitions the computational bitstrings into subspaces according to their hamming weight, with corresponding eigenvalues. We define:
    \[
    \hat{\obs} = \sum_{b \in \{0,1\}^{2^n}} |b|_h \ \ket{b}\bra{b}
    \]
    where $|b|_h$ is the hamming weight of bitstring $b$. This observable is intended when the solution is the size of a subset of the vertices of $\Gamma$ - that is, the expected solution lies in the interval $[0, n]$.
\end{definition}

\begin{definition}[\Crater] Binary hamming weight.
    Here, we partition the computational basis into two eigenspaces with $\pm 1$ eigenvalues:
    \[
    \check{\obs} = \sum_{b \in \{0,1\}^{2^n}}^{|b|_h < \frac{n}{2}} \ \ket{b}\bra{b} 
        \; - \sum_{b \in \{0,1\}^{2^n}}^{|b|_h > \frac{n}{2}} \ \ket{b}\bra{b}
    \]
    This observable naturally reflects a binary classification task, but may also be used for finer grained classification when taking expectation values, as these will lie along the continuous interval $[-1, 1]$.
\end{definition}
Notice that when n is even, the \Crater observable does not span the entire hilbert space. This is to ensure that the $\pm 1$ eigen-subspaces have the same dimension, for an unbiased expectation.

\begin{remark}
    For the \Mountain and \Crater observables, the readout depends only on the hamming weight $|b|_h$ (i.e. on the subset size $|S(b)|$ rather than on which vertices are present). Accordingly, in this setting we do not require the learned representation to preserve the fixed vertex-indicator interpretation used in Bitstring. 
    
    At the level of basis states, the models using \Crater are only trained to place weight on bitstrings of the appropriate sizes, such that the final expected value matches the maximum clique size. They are otherwise free to use any internal basis/representation compatible with the observable.
    In contrast, the \Mountain observable allows less flexibility as the hamming weight of bitstrings is conditioned to encode the clique size directly, although the actual node encoding does not necessarily match the vertices present.
\end{remark}

Note that all the observables considered above are $S_n$-invariant. This is trivial for $O=I$, and the invariance of the \Mountain and \Crater observables follows from the invariance of the hamming weight with respect to permuting bitstrings. 

This can be phrased equivalently at the level of eigenvalues. For \Mountain and \Crater, the eigenvalues are assigned solely by hamming weight, and hence are constant on $S_n$-orbits of computational basis states. Consequently, conjugation by $\phi(\sigma)$ preserves each eigenspace for any $\sigma \in S_n$, leaving the observable unchanged.

\subsection{Loss Functions}

For a given graph $\Gamma$, let $Y(\Gamma)\subseteq\{0,1\}^n$ be the set of target bitstrings and let $\overline{Y}(\Gamma)$ denote its complement. In the idealised setting, the desired output is any state whose measurement support lies entirely on $Y(\Gamma)$. An unbiased target is the uniform superposition over all of $Y(\Gamma)$,
\[
|\psi_\Gamma^\star\rangle := \frac{1}{\sqrt{|Y(\Gamma)|}}
\sum_{b\in Y(\Gamma)} |b\rangle
\]
although the loss functions considered here do not require uniformity over $Y(\Gamma)$. Fix a graph $\Gamma$ and denote the model output state by
\[
\rho_{\theta}(\Gamma) := U(\Gamma,\theta)\,\rho(\Gamma)\,U(\Gamma,\theta)^\dagger.
\]
When measuring in the computational basis, this induces the distribution
\[
\mathds{P}_{\theta}(b \mid \Gamma) := \langle b|\rho_{\theta}(\Gamma)|b\rangle,
\qquad b\in\{0,1\}^n.
\]

\begin{definition}\label{def:distloss}
    We use a \textit{distribution loss} that penalises probability mass assigned to incorrect (non-maximum) outputs, without enforcing any particular distribution over the correct superposition:
    \[
    \mathcal{L}_{\mathrm{dist}}(\Gamma,\theta)
    := \sum_{b \in \overline{Y}(\Gamma)} \mathds P_{\theta}(b \mid \Gamma)^2.
    \]
\end{definition}

Note that the distribution loss equals $0$ whenever the support of $\mathds P_\theta(\cdot\mid\Gamma)$ is contained within $Y(\Gamma)$, regardless of how probability is distributed among the correct bitstrings.

To increase sensitivity when the total wrong mass is small, we also consider a logarithmic variant.

\begin{definition}\label{def:logwrong}
   The \textit{log-wrong distribution loss} is given by 
   \[
   \mathcal{L}_{\mathrm{logwrong}}(\Gamma,\theta)
    := \log\!\Big(\varepsilon + \sum_{b \in \overline{Y}(\Gamma)} \mathds P_{\theta}(b \mid \Gamma)\Big),
    \]
with a small constant $\varepsilon>0$ for numerical stability.
\end{definition}

Finally, for evaluation we report an argmax-based success criterion.  Let
\[
b^\star(\Gamma,\theta) \in \arg\max_{b} \mathds{P}_{\theta}(b\mid\Gamma)
\]
be a most-likely bitstring under the model distribution.

\begin{definition}\label{def:argmaxloss}
     We define the argmax loss
    \[
    \mathcal{L}_{\mathrm{argmax}}(\Gamma,\theta)
    := \mathds{1}\!\left[b^\star(\Gamma,\theta)\notin Y(\Gamma)\right],
    \]
    which is $0$ if the most probable prediction encodes a bitstring in $Y(\Gamma)$ and $1$ otherwise.    
\end{definition}

Note that the loss $\mathcal{L}_{\mathrm{argmax}}$ is non-differentiable and we use it only as an evaluation metric rather than as a training objective. 

For the \Mountain and \Crater observables, we additionally consider expectation-based loss functions.
Write $\trace{\obs(X) \rho_{\theta}(\Gamma)}$ for the model output, and let $y(\Gamma)$ be the target label.

\begin{definition}\label{def:mseloss}
    The mean-squared error (MSE) loss is 
    \[
    \mathcal{L}_{\mathrm{MSE}}(\Gamma,\theta)
    := \left(
        \trace{\obs(X) \rho_{\theta}(\Gamma)}  - y(\Gamma)
    \right)^2
    \]
\end{definition}

\begin{definition}\label{def:mountainloss}
    The \Mountain loss is 
    \[
    \mathcal{L}_{\mathrm{\Mountain}}(\Gamma,\theta)
    := (1 - \alpha) \cdot \mathcal{L}_{\mathrm{MSE}}(\Gamma,\theta) + \alpha \cdot \mathcal{L}_{\mathrm{logwrong}}(\Gamma,\theta)
    \]
    for a hyper-parameter $\alpha$ that controls whether to target the single-shot probability, or correctness of the expectation value.
\end{definition}

If the task contains discrete labels, the output space may be binned, with target interval given by $y_{\mathrm{\pm}}(\Gamma)$, such that $y(\Gamma) \in y_{\pm}(\Gamma)$.
\begin{definition}\label{def:craterloss}
    We define the \Crater loss as the MSE loss with an extra mis-classification penalty:
    \[
    \mathcal{L}_{\mathrm{\Crater}}(\Gamma,\theta)
    := \mathcal{L}_{\mathrm{MSE}}(\Gamma,\theta) + \alpha \cdot \mathds{1}\left[\trace{\obs(X) \rho_{\theta}(\Gamma)}  \notin y_{\pm}(\Gamma)\right]
    \]
\end{definition}

%% file: Paper/experiments/tasks.tex
\section{Tasks}\label{sec:tasks}

In this section, we define the optimisation problems that will be used to illustrate the theoretical framework developed above. In the experiments that follow, we test $S_n$-invariant models constructed from the components defined in the preceding section and show that Semi-Symmetric models exhibit stronger trainability and generalisation compared to unsymmetrised counterparts on these tasks.

\subsection{The Maximum Clique Problem}

Given a simple undirected graph $\Gamma = (V,E)$ with $|V| = n$, a \emph{clique} is a subset
$S \subseteq V$ such that every pair of vertices in $S$ is connected by an edge.
The \emph{maximum clique problem} (Max-Clique) returns a clique of maximum cardinality, denoted $\omega(\Gamma)$, for a given input graph $\Gamma$.

Recall that under the \Bitstring observable (Definition \ref{def:bitstring}), we interpret each computational basis state \(|b\rangle,\; b\in\{0,1\}^n\) as the indicator function of a vertex subset
\(
S(b) := \{\, v\in \Gamma_v : b_v = 1 \,\}.
\) 
A bitstring $b$ is said to encode a clique if $S(b)$ is a clique in $\Gamma$.
Let $\mathrm{MC}(\Gamma)\subseteq\{0,1\}^n$ denote the set of bitstrings whose supports are cliques of maximum cardinality in $\Gamma$. Then the task becomes to develop a heuristic that returns a bitstring $b \in \mathrm{MC}(\Gamma)$ with high probability. 

The decision version of Max-Clique requires determining whether $\Gamma$ contains a clique of size at least $k$ and is NP-complete, while the optimisation problem is NP-hard.
Exact algorithms scale exponentially in the worst case, with best-known classical methods
based on branch-and-bound, backtracking, or integer programming.
Approximation is also hard: unless $\mathrm{P}=\mathrm{NP}$ the Max-Clique problem cannot be approximated within $n^{1-\varepsilon}$ for any $\varepsilon>0$.

Known heuristic classical approaches include greedy local search \cite{pullan2006dynamic}, simulated annealing \cite{GENG20075064},
and evolutionary methods \cite{evolution}.
Quantum approaches studied in the literature include adiabatic optimisation \cite{adiabatic},
QAOA-style encodings, and quantum-inspired tensor network heuristics \cite{dzibuyna}, though none
are known to provide general polynomial-time advantage. 

\subsection{Finding the Size of a Maximum Clique}

A related but strictly weaker task is to determine only the clique number $\omega(\Gamma)$,
without identifying the vertices forming a maximum clique.
This problem remains NP-hard, as it directly solves the decision version via thresholding.

From an algorithmic perspective, this task admits different heuristic strategies:
spectral bounds, semidefinite relaxations, and statistical estimators can sometimes
approximate $\omega(\Gamma)$ without producing a clique itself.
In the quantum setting, this relaxation allows us to define observables (\Mountain and \Crater) that only depend on the cardinality of vertex subsets rather than their identity, enabling models that discard explicit vertex-indicator semantics.

%% file: Paper/theory/pine.tex
\section{Hybrid Algorithms}

Taking inspiration from the structure of classical algorithms, we also consider hybrid approaches to the problem where a quantum heuristic is used to guide a classical search algorithm. Recursive-QAOA (RQAOA) \cite{bravyi_obstacles_2020} is the prototypical example, in which a QAOA sub-routine is used to determine how the overall problem is to be reduced at each step.
Recently, more sophisticated ways of splitting a problem into sub-problems have been considered \cite{brady_iterative_2024, esposito_hybrid_2024, finzgar_quantum-informed_2024, brady_quantum_2025}, in which the reduction choices at each step are directly informed by the problem. While \cite{esposito_hybrid_2024} implements a divide-and-conquer style algorithm, with a quantum component for the re-combination step, \cite{brady_iterative_2024, finzgar_quantum-informed_2024, brady_quantum_2025} consider greedy algorithms, directly generalising RQAOA. We follow the latter approach here.

\subsection{\Pine}
Inspired by the dynamic programming approach, we construct an outer skeleton algorithm \Pine that wraps a quantum heuristic designed to solve a specific sub-task.
In order to keep training costs low, we re-use the optimised parameters across multiple graphs at inference time - this is in contrast to the standard RQAOA approach that typically requires independent optimisation for each problem instance.

\subsubsection{Recursion}
At each step, \Pine accepts a distribution over the nodes in the current graph, which we sample from to obtain the next node to be added to the clique. The graph is then reduced to contain only that node's neighbours. The algorithm then proceeds to find the maximum clique in the resulting subgraph.

We obtain the node distribution as the marginal distribution over each node after applying a graph ansatz. This marginal distribution can either be approximated from samples, or simulated exactly. We form candidate cliques by successively adding vertices, and therefore restrict the recursion step to selecting which node to include next.
The algorithm terminates when the resulting subgraph is empty.

\begin{figure}[h]
    \centering
    \includegraphics[width=0.9\linewidth]{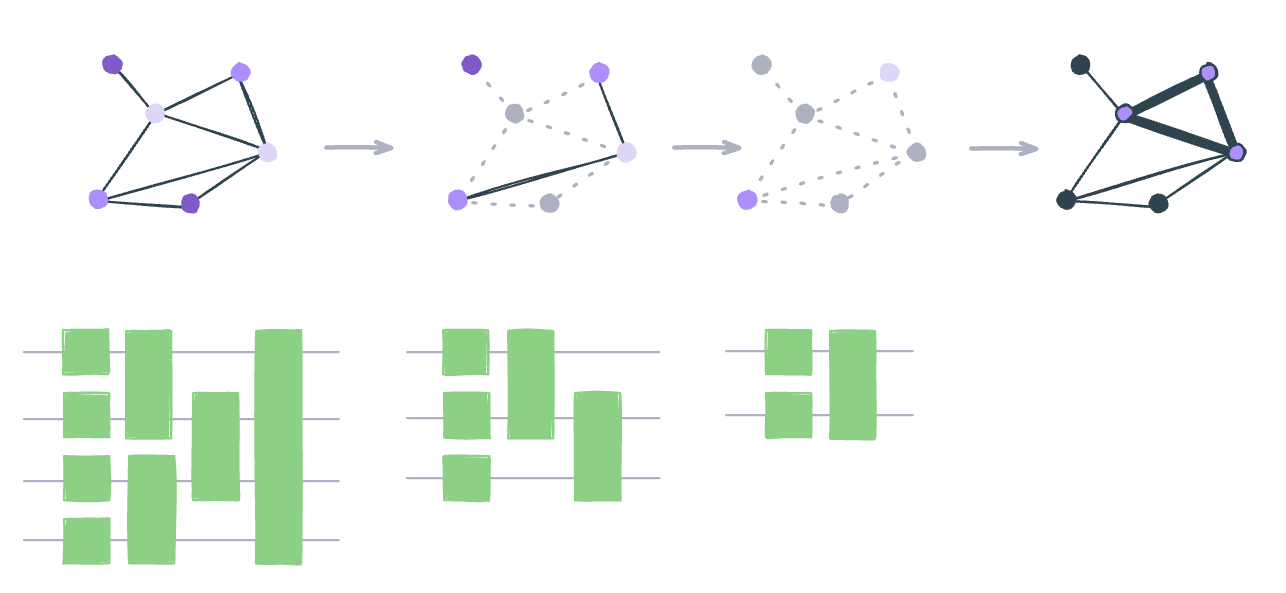}
    \caption{The recursion step of the \Pine algorithm. At each step, the quantum circuit assigns a distribution to the nodes and samples a node from this distribution to add to the clique. A new graph is then constructed from the selected node's neighbours before the algorithm repeats. When the resulting subgraph is empty \Pine terminates and outputs the constructed clique.}
    \label{fig:max-clique-recursion}
\end{figure}

%% file: Paper/experiments/maxcliquemodels.tex
\section{Maximum Clique Models}
We now compare the effects of symmetry and recursion on trainability and generalisation in the context of the max-clique problem. \autoref{tab:model_zoo_mc} summarises the configuration details of the models used. In all cases, the input state is a trainable permutation-invariant product state $|\psi\rangle^{\otimes n}$. Measurement is performed using the \Bitstring observable, and training uses the LogWrong loss.

The initial state and observable are $S_n$-invariant. For the \Rook and \Parliament models, the circuit is $S_n$-equivariant, and so by Proposition \autoref{prop:semi_sym_invariance}, the resulting loss is invariant under vertex relabelling. 

Unlike the \Rook model, the \MF model circuit is not permutation-equivariant (see Remark \ref{rem:MF_not_equivariant}) and therefore not $S_n$-invariant in general.

\begin{table}[h]
    \centering
    \begin{tabular}{l | lr | rrr}
        \textbf{Model} & \textbf{Layers} & \textbf{Params} & \multicolumn{3}{c}{\textbf{Train dataset}}\\
        & & & \textbf{$p$} & \textbf{$n$} & \textbf{Graphs per.}  \\
        \hline
        \hline
        \Rook          & 5    &   38 & 0.1-0.9 (step 0.1) & 2-6, 8 & all, 111\\
                       & 20   &  143 & & & \\
                       \cline{2-6}
                       & 5    &   38 & 0.1-0.9 (step 0.1) & 7,8,9,10 & 44\\
                       &      &      &                    & 12,14,16 & 9\\
        \hline
        \Parliament    & 5    &   38 & 0.5 & 2-16 & 91 \\
        \hline
        \MF & 5, 2 &  118 & 0.1-0.9 (step 0.1) & 2-6, 8 & all, 111 \\
                      & 20, 5& 1063 & & & \\
    \end{tabular}
    \caption{An overview of the models and datasets considered in this work for the Max-Clique task. The parameters $n$ and $p$ refer to the number of nodes and the edge probability that define the Erd\H{o}s--R\'enyi graph distribution.}
    \label{tab:model_zoo_mc}
\end{table}

\begin{table}[h]
    \centering
    \begin{tabular}{lll}
        \textbf{$p$} & \textbf{$n$} & \textbf{Graphs per.}  \\
        \hline
        \hline
        all & 2-6 & 207 (all)\\
        0.5 & 7-16 & 200\\
    \end{tabular}
    \caption{An overview of the testing dataset. Note that for models trained on graphs with 2-6 nodes, the test distribution overlaps with the training over this range, as the number of distinct graph classes in this range is limited.}
    \label{tab:model_zoo_kc}
\end{table}

\begin{remark}
Compared to the \Rook model, the \MF architecture uses significantly more trainable parameters per layer.
The choice of using at least 2 inner layers was motivated by maintaining sufficient expressivity to capture the \Rook ansatz, while the 5 inner layers used for 20 layer model were selected to maximise expressivity, as per \cite{sim_expressibility_2019}.
\end{remark}

\subsection{Training dynamics}\label{subsec:training_dynamics}
For each of the models in \autoref{tab:model_zoo_mc}, we plot the mean training gradient, as well as the loss and selection accuracy (Argmax) trajectories in \autoref{fig:max-clique:grad_summary}. We observe that the \Rook models typically train faster and observe better gradients than the \MF models. Even in the case where the number of parameters is approximately equivalent, we find that the \MF model takes longer to converge.
Moreover, \MF models converged to a lower final accuracy than the equivalent \Rook models trained on the same dataset. 

\begin{figure}[h]
    \centering
    \includegraphics[width=\linewidth]{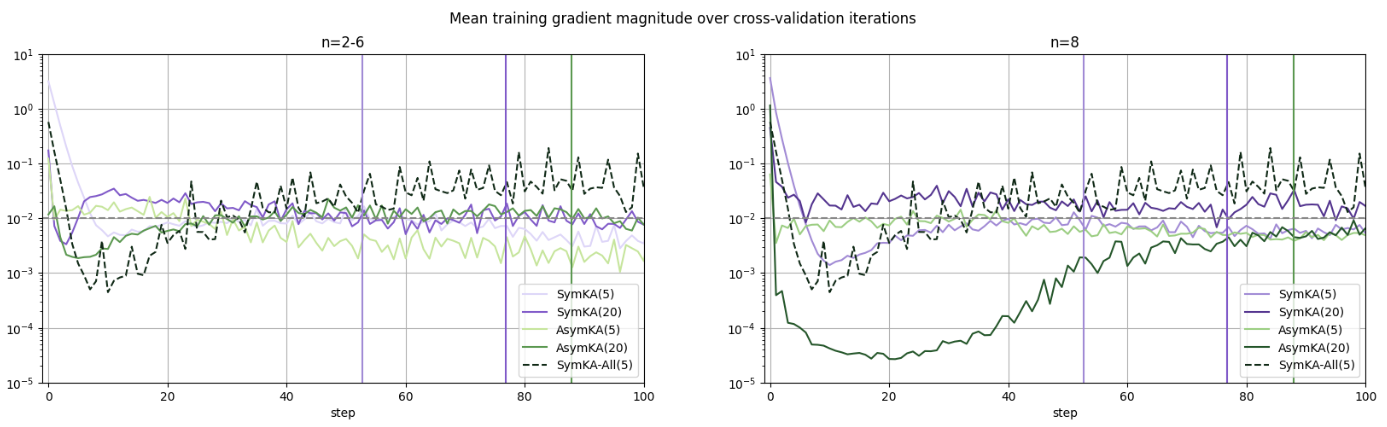}
    \\\vspace{1em}
    \includegraphics[width=\linewidth]{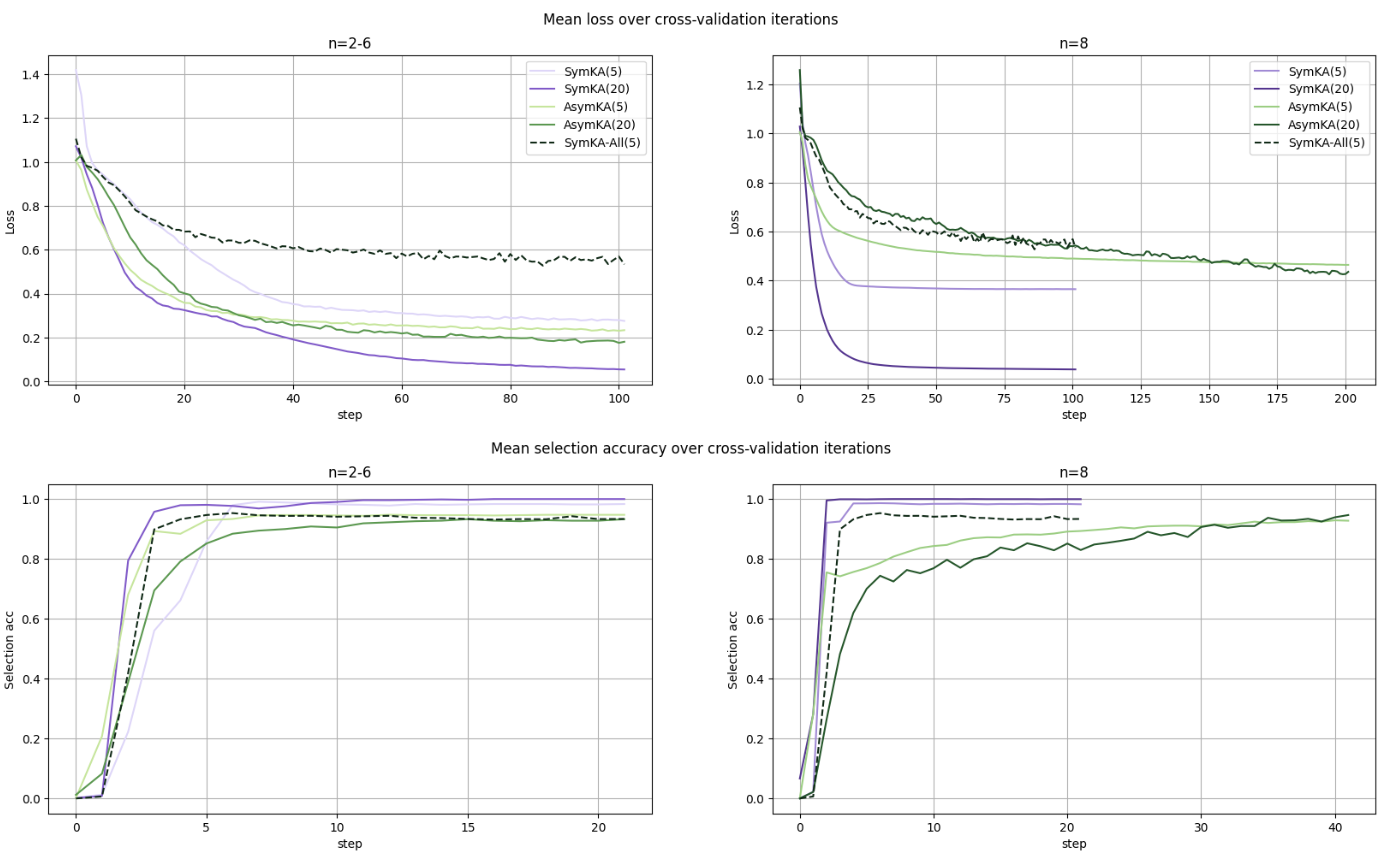}
    
    \caption{Average gradient magnitude over the training trajectory. Vertical lines show the average number of steps required to obtain the best model. Note that we did not evaluate the selection accuracy at every gradient step.}
    \label{fig:max-clique:grad_summary}
\end{figure}

\subsection{Generalisation capacity}\label{subsec:results-mc-generalisation}
Next, we consider the capacity for models to generalise to graphs from unseen distributions. The results are plotted in \autoref{fig:max-clique:acc_summary} for non-recursive \Rook and \MF models.

\subsubsection{\Rook and \MF models}
\begin{figure}[h]
    \centering
    \includegraphics[width=\linewidth]{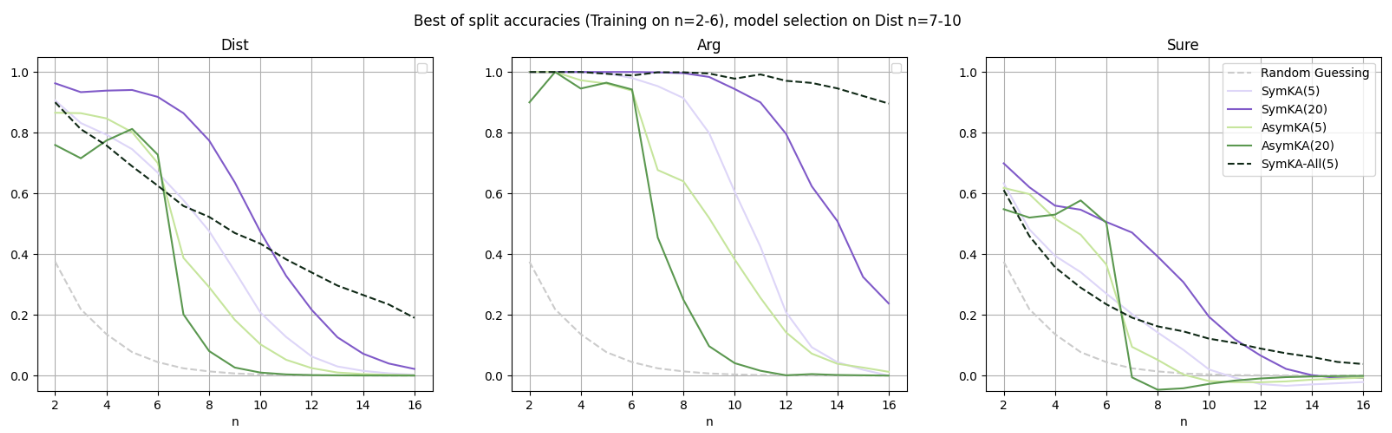}
    \\\vspace{1em}
    \includegraphics[width=\linewidth]{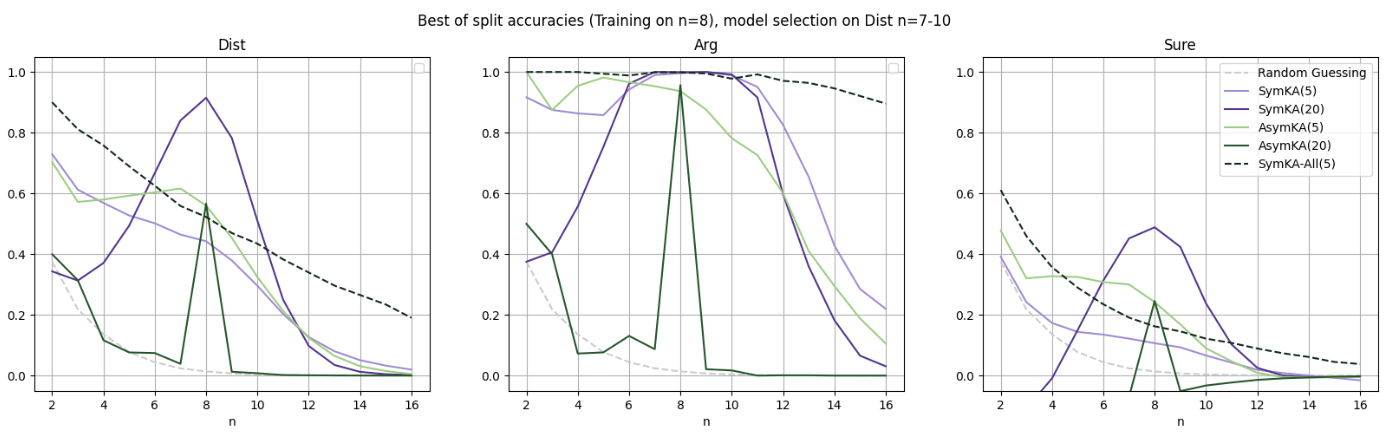}
    \caption{Average model performance over 5-fold cross-validation runs, with at least 2 iterations per split.}
    \label{fig:max-clique:acc_summary}
\end{figure}

We observe that the equivariant models tend to both display less overfitting over graphs of the sizes they were trained on, as well as a typically higher accuracy centered on the distribution closest to the training distribution.

Interestingly, while increasing the number of internal layers in \Rook models improves the accuracy of the model on its training data, this does not necessarily translate to improved generalisation. Indeed, we find that the training setup favours early stopping for the 5 layer \Rook, as the Argmax accuracy converges much faster than the distribution accuracy. While this favours generalisation, we find that this lets the \Rook model fall behind the 5 layer \MF model on the same dataset.

Conversely, for the \MF models we find that adding layers leads to worse performance overall, as models overfit on the training data without improvements to accuracy. This is likely due to the observed training dynamics above, making it difficult for the model to navigate towards good solutions.
The 20 layer \MF model exhibits the steepest decline, achieving no better than the random guessing baseline over unseen graph sizes.
In almost all cases (except the 5 layer \Rook model trained on 8 node graphs), we observe that \Rook models achieve better generalisation than comparable \MF models.
In \autoref{subsec:results-mc-equivariance}, we additionally investigate whether the \MF models are able to learn the symmetries of the data, and find that they do not.

We include a \Parliament model in order to verify whether using a single set of parameters for the full set of graphs is viable. Indeed, we find that the distribution accuracy appears to decay linearly (in contrast to the exponential decay in random guessing), while the Argmax accuracy remains high throughout.

\subsubsection{Fixed-angle generalisation}
For certain families of regular graphs, there exists a fixed-angle conjecture for QAOA circuits, where parameters that have been optimised on a single graph instance can be generalised without further training to other graphs with similar properties \cite{wurtz_fixed_2021, wybo_missing_2025}.
These results do not apply directly in our setting, as we consider Erd\H{o}s–R\'enyi graphs with various edge probabilities $p$ and node sizes $n$. For these graphs, the expected degree of each vertex is $np$ and thus not fixed.
Additionally, the fixed angles are introduced primarily as a result of locality of the QAOA algorithm at fixed depth $L$.
Indeed, Max Clique (and Max Independent Set) exhibit the overlap-gap property, which places an asymptotic bound on the solution quality that can be achieved via local algorithms \cite{farhi_quantum_2020}.
In this work we consider an ansatz over complete graphs, which breaks the locality assumption.

From the \Parliament models, we identify the fixed parameters that achieve the best average performance over graph sizes from 2 to 16 nodes.
We additionally demonstrate that models can converge towards these parameters by training on a set of smaller graphs, and hence at a much lower cost. In \autoref{fig:parliament_approx}, a visualisation is given for the best fixed angle models obtained during training. Though not all models trained this way converge to these optimal parameters, the fact that we can train models at all suggests that the same technique could be applied to find models at larger scales where the optimal parameters cannot be trained for directly.

\begin{figure}[h]
    \centering
    \includegraphics[width=\linewidth]{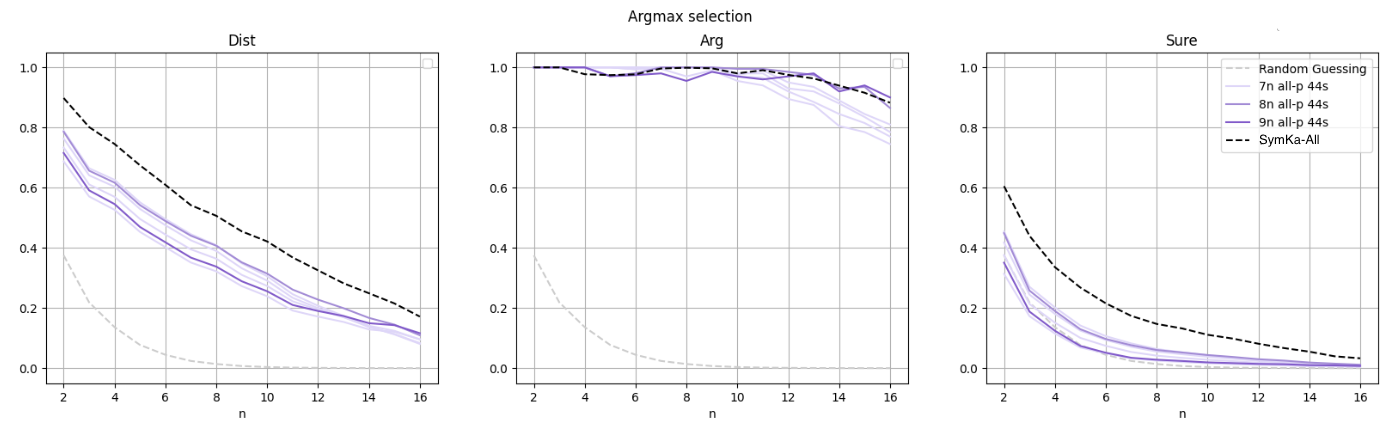}
    \caption{Models reaching \Parliament-like generalisation curves via training on graphs of a single size only.}
    \label{fig:parliament_approx}
\end{figure}

\subsubsection{Model extrapolation}
Alongside the generalisation modes, we also find that by using a different criterion for model selection (SurenessArgmax), the resulting models perform very well on graphs of a fixed size. Moreover, the optimal parameters for the different graph sizes are all similar. We hence investigate a warm-start scheme for training models on graphs of size $n+m$ using the parameters optimised from graphs of size $n$. Indeed, as demonstrated in \autoref{fig:macaw}, provided the step size $m$ is sufficiently small, we find that this allows for fast convergence with large gradient magnitudes that do not decay with the system size.

While the model accuracy tends towards the \Parliament model performance, the training process is overall more achievable. Thus, while training a 20 layer \Parliament model would be infeasible, it could be approximated by a series of warm-started models.

\begin{figure}[h]
    \centering
    \includegraphics[width=\linewidth]{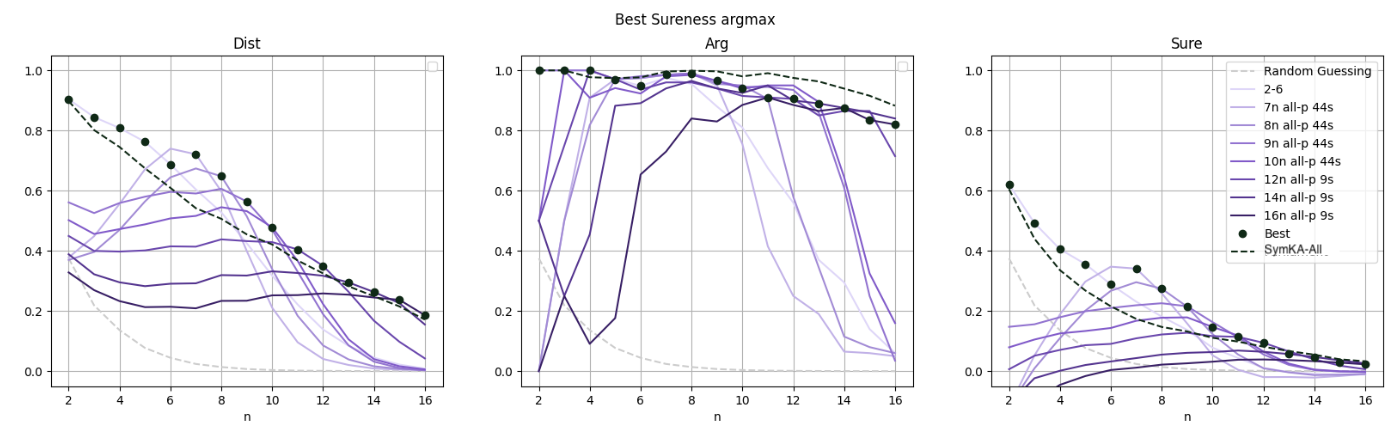}
    \\\vspace{1em}
    \includegraphics[width=\linewidth]{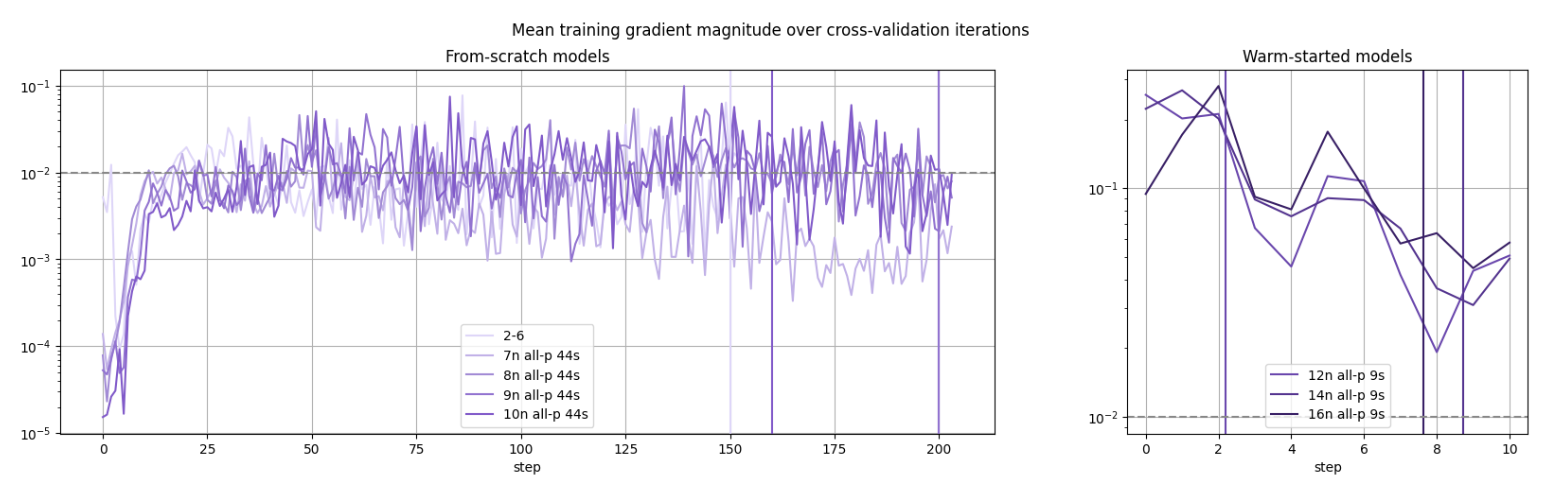}
    \caption{The best single graph-size model accuracies. Each line represents a single model. For models trained on graph sizes with $n = 12,14,16$, the parameters were warm started from a model trained on $n-2$ sized graphs.}
    \label{fig:macaw}
\end{figure}

\subsubsection{\Pine models}
\begin{figure}[h]
    \centering
    \includegraphics[width=0.49\linewidth]{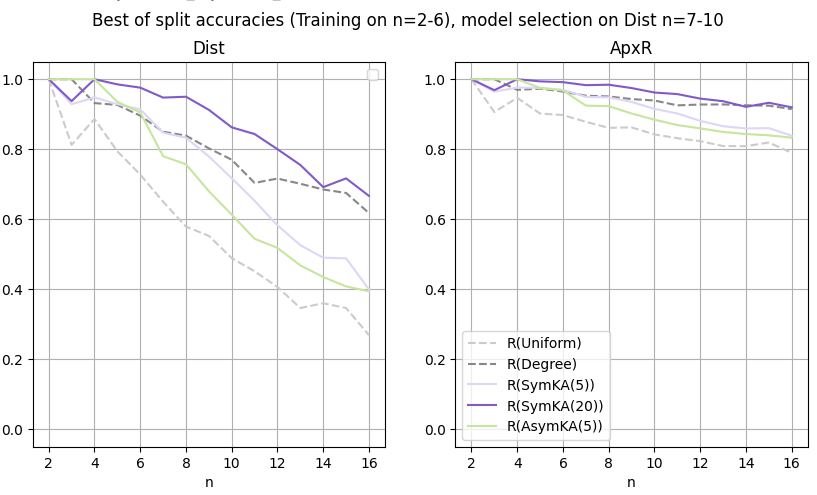}
    \hfill
    \includegraphics[width=0.49\linewidth]{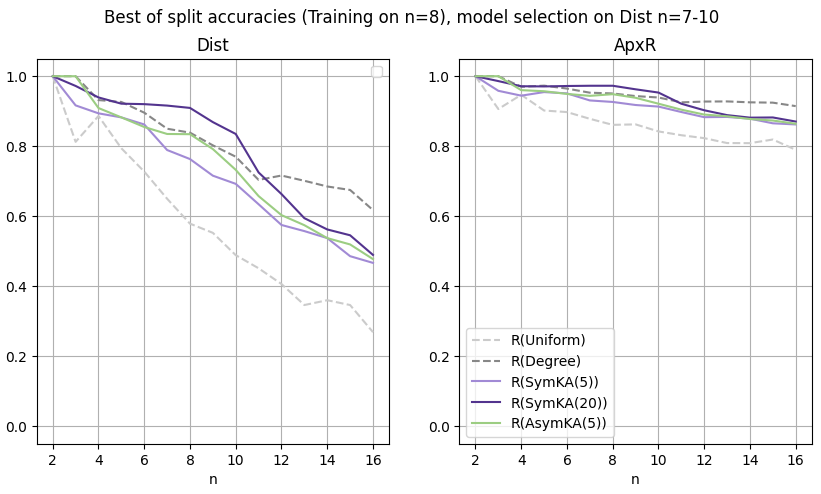}
    \caption{Average \Pine model performance over 5-fold cross-validation runs, with at least 2 iterations per split.}
    \label{fig:max-clique:acc_summary_pine}
\end{figure}
For \Pine models, as per \autoref{fig:max-clique:acc_summary_pine}, we observe that the single-shot (distribution) accuracy is significantly improved, and indeed is higher compared to any choice of inner (non-recursive) model.
This improvement is in part derived from the fact that \Pine ensures that the output solution is guaranteed to be a clique, unlike when sampling from the raw circuit distribution.

Specifically, instead of requiring the circuit to place high probability mass on an entire maximum clique at once, the \Pine procedure only needs to consistently select vertices that remain compatible with \textit{some} maximum clique. This leads to a multiplicative improvement in success probability across steps, as errors are filtered out early and the problem size shrinks, concentrating probability mass onto increasingly clique-dense subgraphs.

The hybrid nature of \Pine also allows better allocation of the same amount of quantum resources - given a certain budget of quantum computation time, executing a series of smaller circuits interspersed with classical computations leads to better results than executing a larger quantum circuit to solve the problem in a single step. This flexibility is particularly well-suited to the constraints of Noisy Intermediate-Scale Quantum (NISQ) era devices.

Encouragingly, we also find that the quantum heuristic is competitive relative to classical baselines, although it only surpasses the upper baseline in the range where the inner model has a high distribution accuracy.
We consider two classical heuristics as drop-in replacements for the quantum component in the \Pine algorithm.
\Uniform outputs a uniform distribution over the nodes. This baseline thus acts as a lower bound, and quantifies the improvement obtained from the \Pine method in isolation.
The second heuristic, \Degree, selects vertices with probability proportional to their degree. 

%% file: Paper/experiments/findkmodels.tex
\section{Finding Maximum Clique Size Models}

Next, we turn to the task of only identifying the size of the largest clique (the clique-number), without needing to supply the clique itself.
Notice that this task is comparatively easier than the Max-Clique task, while remaining NP-Hard. This is illustrated by the random guessing baseline, which is significantly higher for the Clique-Number task. The baseline is computed as the probability of being correct when sampling a random bitstring and reducing it via the \Mountain observable.
\begin{table}[h]
    \centering
    \begin{tabular}{l | lr | rrr}
        \textbf{Model} & \textbf{Layers} & \textbf{Params} & \multicolumn{3}{c}{\textbf{Train dataset}}\\
        & & & \textbf{$p$} & \textbf{$n$} & \textbf{Graphs per.}  \\
        \hline
        \hline
        \Mountain-\Rook & 5    &   38 & 0.1-0.9 (step 0.1) & 8 &  111\\
        \Crater-\Rook   & 5    &   38 & 0.1-0.9 (step 0.1) & 8 &  111\\
        \hline
        \Mountain-\MF   & 5, 2 &  118 & 0.1-0.9 (step 0.1) & 8 &  111\\
        \Crater-\MF     & 5, 2 &  118 & 0.1-0.9 (step 0.1) & 8 &  111\\
    \end{tabular}
    \caption{An overview of the models and datasets considered in this work for the Clique-Number task. The parameters $n$ and $p$ refer to the number of nodes and the edge probability that define the Erd\H{o}s--R\'enyi graph distribution.}
    \label{tab:model_zoo_kc}
\end{table}

\subsubsection{\Mountain Models}
The \Mountain models uses either \Rook or \MF circuits with 5 layers (2 inner layers for the \MF model), the permutation-invariant initial random product state $|\psi\rangle^{\otimes n}$, the \Mountain observable, and \Mountain loss with $\alpha = 0.5$.

The initial state and observable are invariant, while the circuit is $S_n$-equivariant in the \Rook case unlike the \MF case.
Loss-invariance thus follows from Proposition \autoref{prop:semi_sym_invariance} for the \Rook model alone.
The training data characteristics are summarised in \autoref{tab:model_zoo_kc}.

\subsubsection{\Crater Models}
Similarly, we consider both \Rook and \MF circuit configurations, as detailed in \autoref{tab:model_zoo_kc}. In both cases, we use the permutation-invariant random initial state $|\psi\rangle^{\otimes n}$, the \Crater observable, and the \Crater loss with $\alpha = 0.5$.

As with the previous cases, the initial state and observable are invariant, while the circuit is equivariant for the \Rook model and non-equivariant for \MF. Loss-invariance again follows from Proposition \autoref{prop:semi_sym_invariance} for \Rook.

\subsection{Training dynamics}
Analogous to the Max-Clique models, we find that the \MF models typically exhibit smaller gradients and converge to worse accuracies over the training iterations. These results are displayed in \autoref{fig:k-clique:grad_summary}.
\begin{figure}[h]
    \centering
    \includegraphics[width=\linewidth]{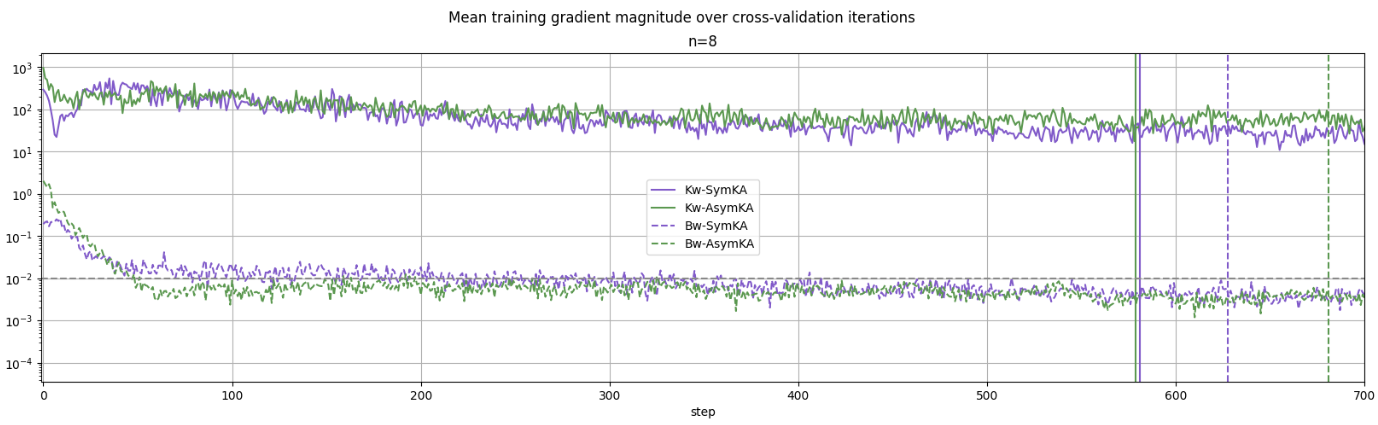}
    
    \caption{Average gradient magnitude, loss value and selection accuracy over the training trajectory. Vertical lines show the average number of steps required to obtain the best model.}
    \label{fig:k-clique:grad_summary}
\end{figure}

\subsection{Generalisation capacity}
We find that the Clique-number models tend to fit best to data-points drawn from a distribution with the same number of nodes as seen in training, following the Max-Clique model behaviour.
The \Crater models exhibit a less smooth generalisation curve, as the eigenvalues of the observable differ between odd and even graph sizes. However, we find that the \Crater-\Rook models generalise better than the \Mountain models particularly on even graph sizes. Meanwhile the \Crater-\MF models significantly overfit the training dataset.

\begin{figure}[h]
    \centering
    \includegraphics[width=\linewidth]{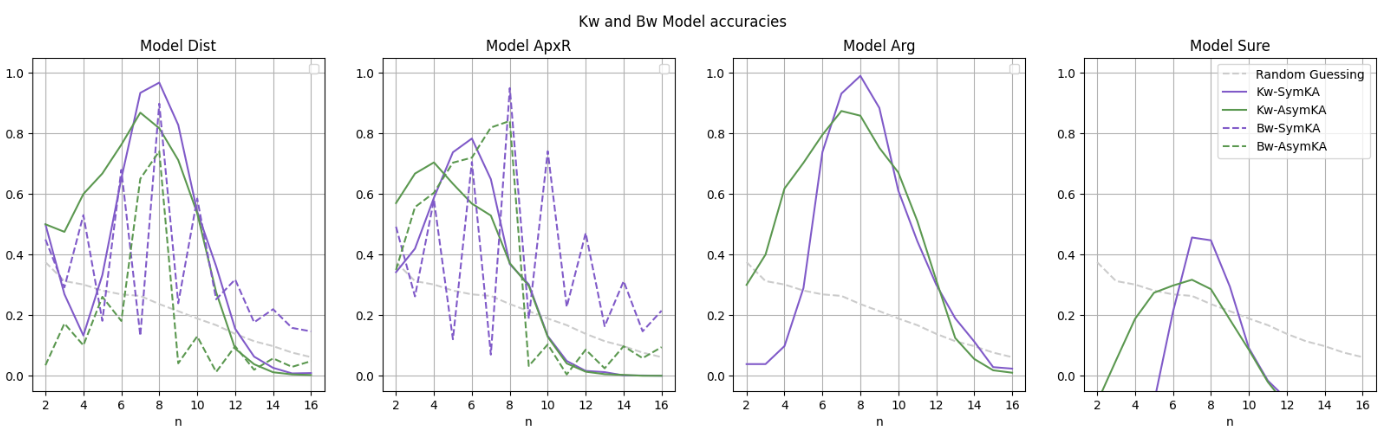}
    \caption{Average model performance over 5-fold cross-validation runs, with at least 2 iterations per split. Note that for Crater models, the Argmax and Sureness accuracies are not defined. The Distribution and Approximation ratio accuracies are computed relative to the expectation values.}
    \label{fig:k-clique:acc_summary}
\end{figure}

\subsubsection{Observable shifts}
While the model's accuracy degrades, by inspecting the raw output distributions we note that the model nevertheless manages to distinguish between graphs with distinct clique sizes.
Indeed, as per \autoref{fig:expectation_drift}, we see that the mean expectation values for a given clique size undergo an almost linear shift as $n$ increases.
While the exact shift depends on the trained model, its reliable shape allows us to construct a `finetuned' model that applies a learnt shift as an additional post-processing step.

To infer the shift, we take a subset of data to be used as a \textit{compositional-validation} set. Using these extra points, we fit the weights on the observable in order to better align the output distributions. The results are plotted in \autoref{fig:kexp_shifted_summary}.
While the regression does not perfectly align, it allows us to extend the generalisation capacity of the model with only a few extra evaluations required. As per \autoref{fig:kexp_shifted_summary}, even a single sample for each node size is sufficient signal to provide a significant improvement.
Indeed, for \Mountain-\Rook models, this is sufficient to match the best possible fit, computed over the entire test data, while for the \Mountain-\MF models the accuracy does not improve as much.

We find that \Crater models were less likely to generalise well, and moreover required different fits depending on the parity of the graph size, as well as additional datapoints in order to approximate the best fit. There was also greater variation in the generalisation potential of the \Crater-\Rook models: while some models managed to maintain distinguishable expectations for each value of $k$, many did not extend beyond a short window around the training dataset, limiting the effectiveness of the fit.

\begin{figure}[h]
    \centering
    \includegraphics[width=0.8\linewidth]{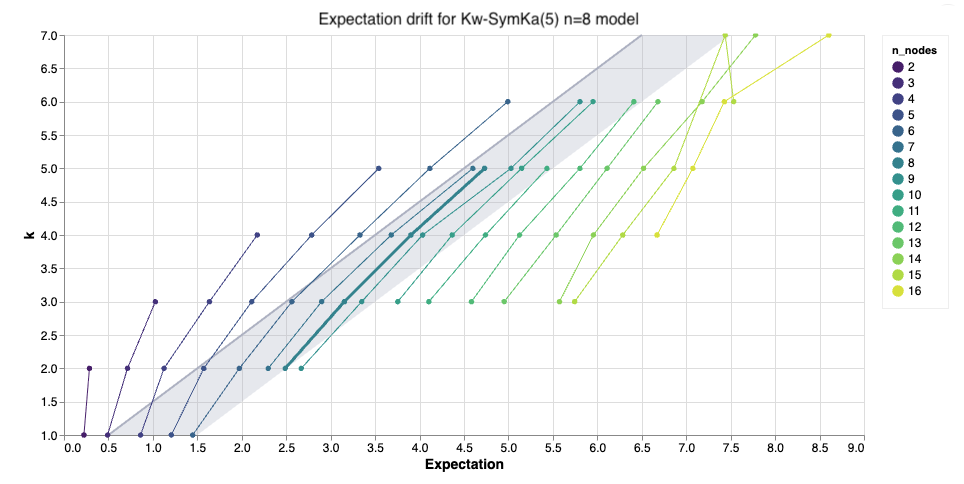}
    \caption{The drift in expected values for a \Mountain-\Rook model as the number of nodes change. We plot the mean expectation value of the model for graphs of a given node and max-clique size. The target $k$ value is on the \textbf{y}-axis, while the predicted value is on the \textbf{x}-axis. Points that fall within the band classify the graphs correctly. Graphs of size seen in the training distribution are shown with a bold line.
    }
    \label{fig:expectation_drift}
\end{figure}

\begin{figure}[h]
    \centering
    \includegraphics[width=0.8\linewidth]{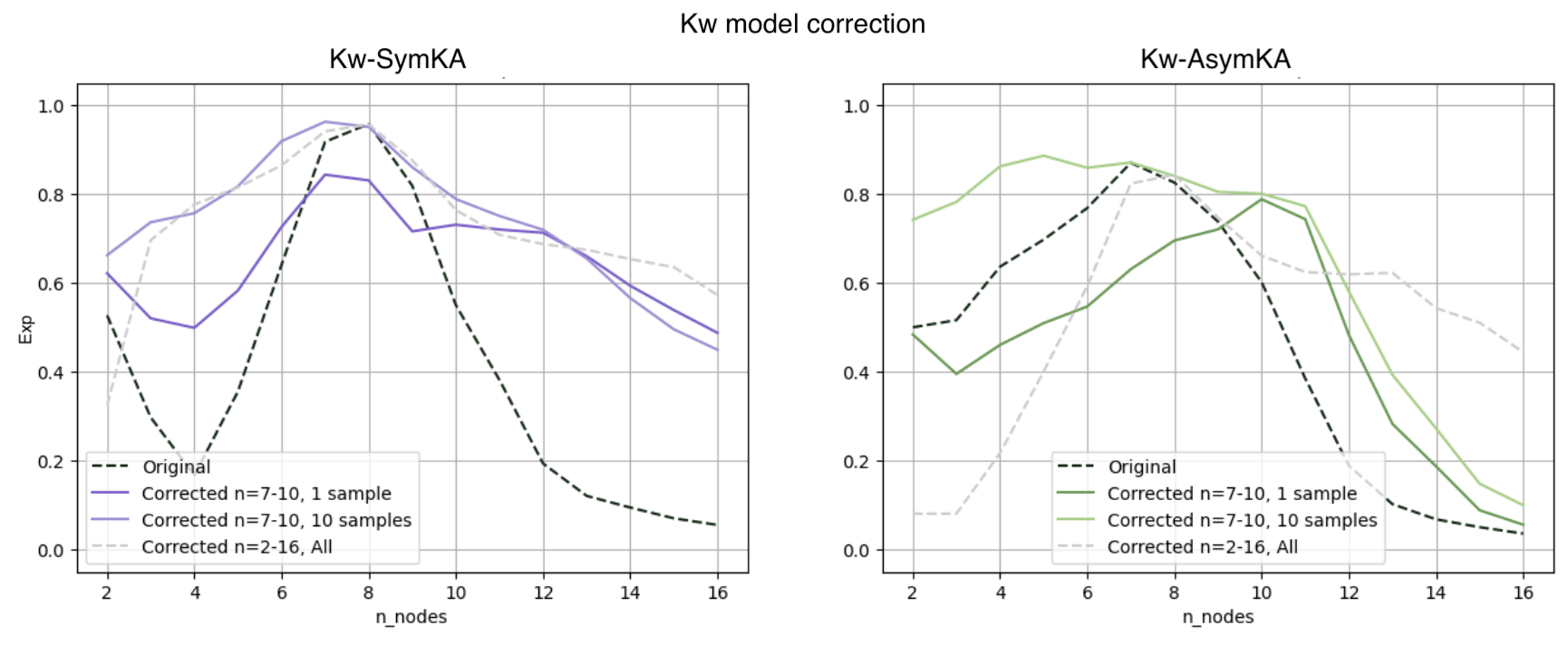}
    \\\vspace{1em}
    \includegraphics[width=0.8\linewidth]{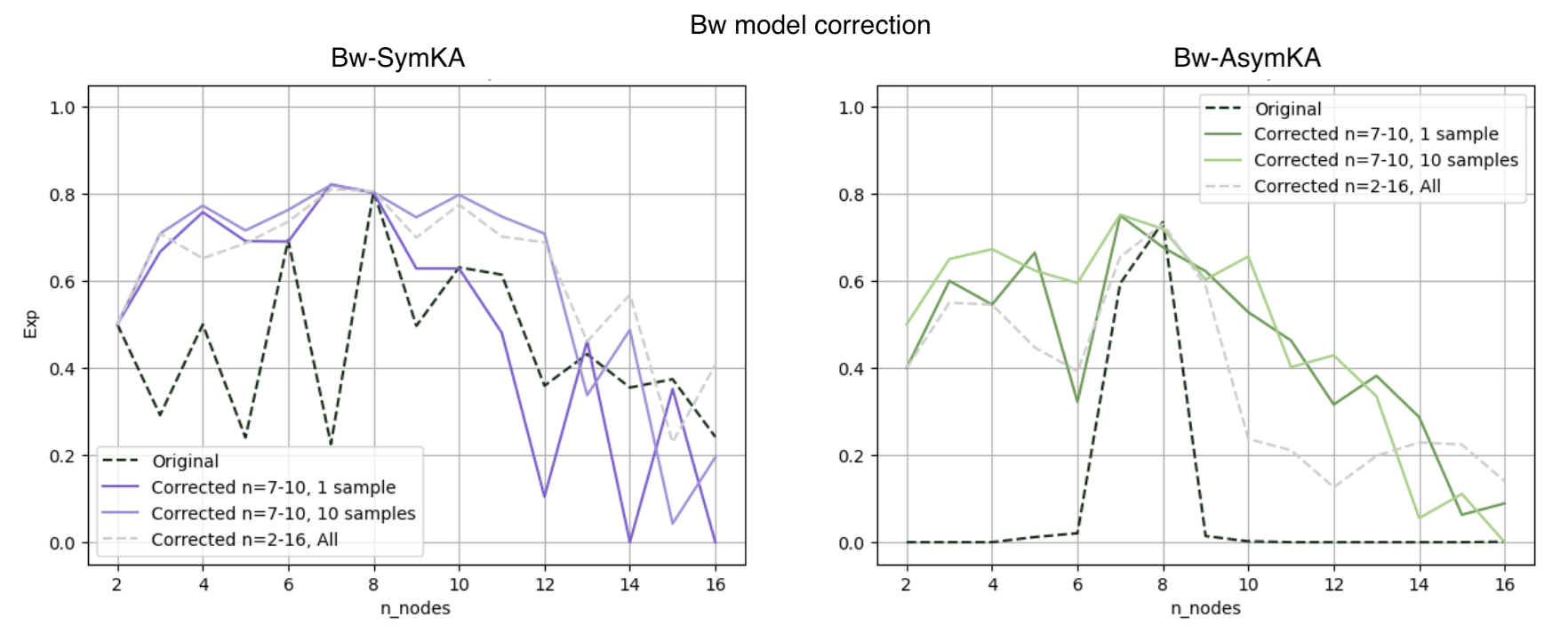}
    \\\vspace{1em}
    \includegraphics[width=0.8\linewidth]{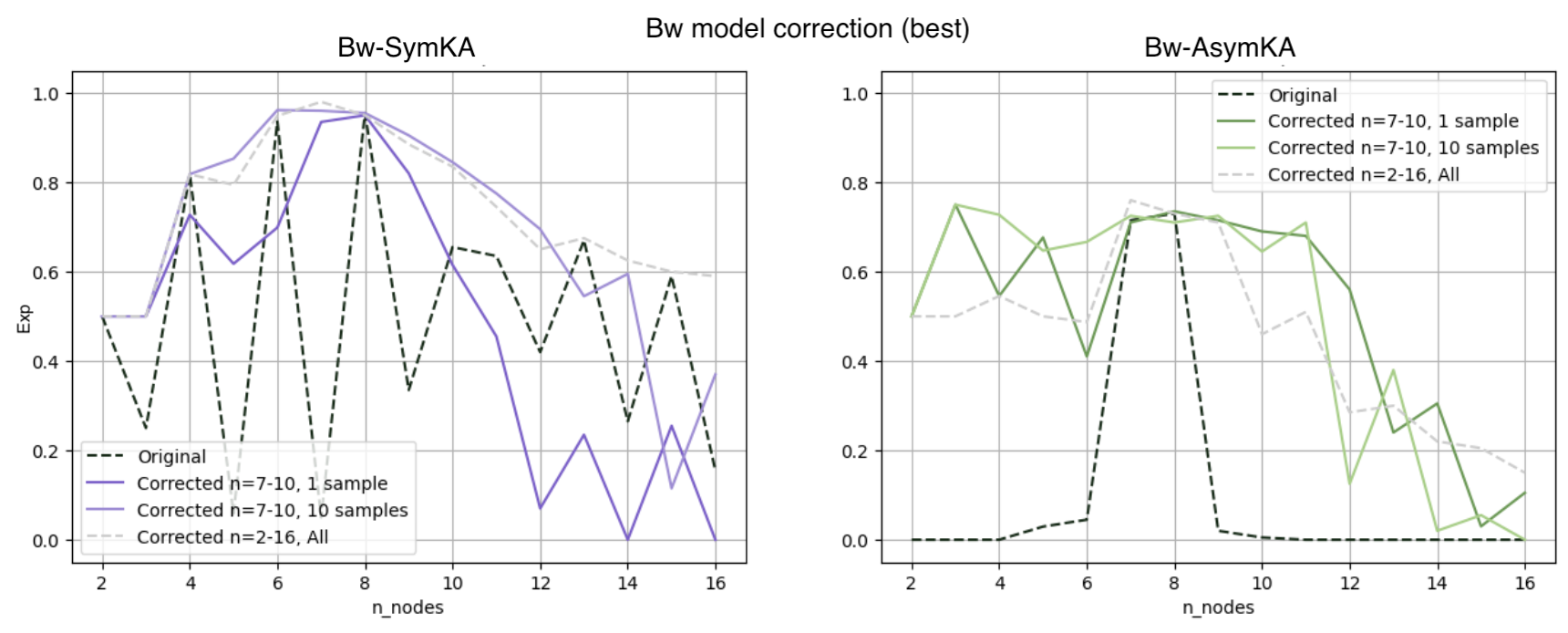}
    \caption{
        Average model performance over 5-fold cross-validation runs, with at least 2 iterations per split.
        We plot the original and corrected accuracies. Exp is the accuracy of the expectation value.
        We include as reference points the best possible fit, obtained via regression on the entire test dataset, and the original accuracy.
        For \Crater models, we additionally plot the best model fit, as there was more variation among the models.
    }
    \label{fig:kexp_shifted_summary}
\end{figure}

%% file: Paper/experiments/conclusion.tex
\section{Conclusion}\label{sec:conclusion}
In this work we described a general Semi-Symmetric framework for constructing group-invariant loss functions using parametrised quantum circuits which avoids the impractical constraint of model constructions without awareness of individual data entries that is necessary for full invariance. Furthermore, we demonstrated that these Semi-Symmetric models retain the trainability and generalisation benefits inherited from symmetry-induced inductive bias. 

We used this framework to design the permutation-equivariant \Rook ansatz for the Max-Clique problem. The \Rook ansatz achieves comparable or superior accuracy to the \MF architecture while using significantly fewer trainable parameters per layer. Also, by including anti-edge layers, the \Rook architecture remains expressive enough to encode optimal solutions for graphs with sparse or empty edge sets, addressing a known locality limitation of QAOA-style encodings.

In our experiments, the Semi-Symmetric models constructed using the \Rook ansatz exhibited consistently larger gradient magnitudes than non-equivariant counterparts, providing empirical evidence that symmetry-constrained loss functions - rather than stricter, fully invariant architectures - can mitigate gradient suppression during training. Correspondingly, we observed that the \Rook models consistently converged faster and to higher final accuracies than \MF models trained on identical datasets, with the gap widening as the number of layers increases. Notably, deeper \Rook models retained stable training dynamics, whereas increasing the number of layers in \MF models lead to overfitting and degraded generalisation.

Moreover, our experiments showed that models trained on small graph sizes (2–6 or 8 nodes) generalised meaningfully to graphs up to 16 nodes, with equivariant models showing less overfitting and higher accuracy on out-of-distribution graph sizes than unsymmetrised versions. These results indicate that symmetry-induced bias provides a more sample-efficient route to generalisation than simply increasing training data volume. 
For symmetrised models, we found that a small number of training restarts is sufficient to obtain models that generalise with a similar profile to the \Parliament models, suggesting that only a small overhead in model selection is required to obtain models with comparable generalisation accuracy to those trained on the entire distribution. In practice, a single fixed set of parameters trained on small instances can be reused at inference time across substantially larger graphs, demonstrating that the compositional approach does not require per-instance retraining.

For the simpler task of clique-number estimation, we observed  that compositional generalisation can be extended further through lightweight post-hoc calibration. In particular, a systematic drift in output expectations across graph sizes can be corrected with as few as one sample per node size, suggesting that the gap between training and deployment distributions need not be a fundamental barrier to scalability.

Finally, our results verified that the \Pine heuristic achieves higher single-shot clique-finding accuracy than any non-recursive inner model alone. This improvement is partially attributable to the progressive reduction in problem size at each recursion step by decomposing the max-clique search into a sequence of smaller sub-problems. However, \Pine also allows a fixed quantum resource budget to be deployed more effectively than a single large circuit evaluation, making it particularly well-suited to NISQ-era hardware constraints. Promisingly, we found that the \Pine quantum heuristic performance is competitive with equivalent classical degree-based greedy baselines.

\subsection{Outlook}

The results in this work support the broader hypothesis that compositional circuits - assembled from small, trainable subcircuits - offer a viable pathway toward scalable quantum learning models. In addition, our results show that hybrid and recursive formulations of quantum heuristics can provide measurable improvements in scalability and inference accuracy. In future work, we will investigate how the branching structure of these recursive procedures can itself be exploited to further improve training dynamics, and explore models in which the internal heuristics are learned jointly within the hybrid structure rather than fixed in advance.

%% file: Paper/appendix/A.tex
\section{A. Group Representations}\label{app:background}

\subsection{Group Representations}\label{subsec:groups}

Assume $\mathbb{k}$ is a characteristic zero field, and let $G$ be a group.

\begin{definition}
A \textit{representation} of $G$ (over $\mathbb{k}$) is a pair $(\pi, V)$ where:
\begin{itemize}
    \item $V$ is a vector space over $\mathbb{k}$
    \item $\pi: G \to \mathrm{GL}(V)$ is a group homomorphism. 
\end{itemize}
\end{definition}

The map $\pi$ defines an action of $G$ on $V$ by: 
\[g \cdot v := \pi(g)v.\]

Recall, the symmetric group $S_n$ is the group of all bijections (permutations) from the set $[n]:= \{1,...,n\}$ to itself. 

\begin{example}\label{ex:sn_perm_rep}Let $G=S_n$, and let $V=\mathbb{k}^n$. Define a representation $S_n \to \mathrm{GL}(V)$ via
\[
\sigma \cdot \left(x_1 e_1+\ldots+x_n e_n\right)=x_1 e_{\sigma(1)}+\ldots+x_n e_{\sigma(n)}
\]
for each $\sigma \in S_n$ where $e_1, ..., e_n$ is the standard basis of $V$. This is called the \textit{permutation representation} of $S_n$.
\end{example}

Generally, given a group action $G \times \mathcal{X} \to \mathcal{X}$ on a set, each group element $g \in G$ permutes the set $\mathcal{X}$ by
\[
g \cdot X = X^\prime \in \mathcal{X}.
\]
The representation in Example \ref{ex:sn_perm_rep} can be viewed as a special case with $\mathcal{X}= \{ e_1, \cdots e_n\}$. Just as in the example, a group action on $\mathcal{X}$ can be extended linearly to a representation on $\mathbb{k}[\mathcal{X}]$. We will refer to this as the permutation representation of $G$ on $\mathcal{X}$.

\begin{definition}
    Let $(V,\pi)$ and $(W,\pi)$ be representations of $G$. A map $f:V\to W$ is said to be $G$-equivariant if
\[
f(\pi(g)v)=\pi_W(g)f(v)\qquad \forall\, g\in G,\ \forall\, v\in V.
\]

Equivalently, for every $g\in G$, the following diagram commutes
\[
\begin{tikzcd}[
    row sep=2.5em,
    column sep=3.5em,
    nodes={inner sep=1.0ex},   
    arrows={shorten <= 0.3em, shorten >= 0.3em}
]
V \arrow[r,"f"] \arrow[d,"\pi(g)"'] & W \arrow[d,"\pi(g)"] \\
V \arrow[r,"f"'] & W
\end{tikzcd}
\]
\end{definition}

Note that linear $G$-equivariant maps are precisely the intertwining operators, or morphisms, in the category of representations of $G$ over $\mathbb{k}$.

\begin{definition}
    Let $(V,\pi)$ be a representation of $G$. A map $f : V \to W$ is $G$-invariant if
\[
f(\pi(g)v) = f(v)
\quad \text{for all } g \in G,\ v \in V.
\]

Equivalently, for every $g \in G$, the following diagram commutes:
\[
\begin{tikzcd}[
    row sep=2.5em,
    column sep=3.5em,
    nodes={inner sep=1.0ex},   
    arrows={shorten <= 0.3em, shorten >= 0.3em}
]
V \arrow[r,"f"] \arrow[d,"\pi(g)"'] & W \\
V \arrow[ur,"f"'] 
\end{tikzcd}
\]
\end{definition}

Note that $G$-invariant maps are simply $G$-equivariant operators between $(V,\rho)$ and the trivial representation of $G$ on $W$.

\subsection{Proofs of stated results}

\begin{proof}[Proof of Proposition \ref{prop:semi_sym_invariance}]

    Case 1: Suppose the circuit $\pqc$ is $G$-equivariant, while $\obs$ and $\rho$ are $G$-invariant. Then we have
    \begin{align*}
        \mathcal{L}(g \cdot X,\boldsymbol{\theta})& =\trace{\obs(X) \bf{U}( g \cdot X, \theta) \rho(X) \bf{U}(X, \theta)^{\dagger}} \\
        &= \trace{\obs(X) \phi(g) U(X, \theta) \phi(g)^{\dagger} \rho(X) \phi(g) \bf{U}(X, \theta) \phi(g)^{\dagger}}\\
        &= \trace{\obs(X) \phi(g) \bf{U}(X, \theta) \rho(X) \bf{U}(X, \theta) \phi(g)^{\dagger}} \\
        &= \trace{ \phi(g)^{\dagger}\obs(X) \phi(g) \bf{U}(X, \theta) \rho(X) \bf{U}(X, \theta) }\\
        &= \trace{ \obs(X) \bf{U}(X, \theta) \rho(X) \bf{U}(X, \theta) } \\
        &=\mathcal{L}( X,\boldsymbol{\theta}).
    \end{align*}

Case 2: Assume the observable $\obs$ is $G$-equivariant while the circuit and state are 
$G$-invariant. Then
\begin{align*}
\mathcal{L}(g \cdot X,\boldsymbol{\theta})
&=\trace{\obs(g \cdot X)\bf{U}(X,\theta)\rho(X)\bf{U}( X,\theta)^{\dagger}} \\
&=\trace{\phi(g)\obs(X)\phi(g)^{\dagger}\bf{U}(X,\theta)\rho(X)\bf{U}(X,\theta)^{\dagger}} \\
&=\trace{\obs(X)\phi(g)^{\dagger}\bf{U}(X,\theta)\rho(X)\bf{U}(X,\theta)^{\dagger}\phi(g)} \\
&=\trace{\obs(X)\bf{U}(X,\theta)\phi(g)^{\dagger}\rho(X)\phi(g)\bf{U}(X,\theta)^{\dagger}} \\
&=\trace{\obs(X)\bf{U}(X,\theta)\rho(X)\bf{U}(X,\theta)^{\dagger}} \\
&=\mathcal{L}(X,\boldsymbol{\theta}).
\end{align*}

Case 3: Assume the circuit and the observable are 
$G$-equivariant, while the state is $G$-invariant. Then
\begin{align*}
\mathcal{L}(g \cdot X,\boldsymbol{\theta})
&=\trace{\obs(g \cdot X)\bf{U}(g \cdot X,\theta)\rho( X)\bf{U}(g \cdot X,\theta)^{\dagger}} \\
&=\trace{\phi(g)\obs(X)\phi(g)^{\dagger}\phi(g)\bf{U}(X,\theta)\phi(g)^{\dagger}\rho(X)\phi(g)\bf{U}(X,\theta)^{\dagger}\phi(g)^{\dagger}} \\
&=\trace{\phi(g)\obs(X)\bf{U}(X,\theta)\phi(g)^{\dagger}\rho(X)\phi(g)\bf{U}(X,\theta)^{\dagger}\phi(g)^{\dagger}} \\
&=\trace{\obs(X)\bf{U}(X,\theta)\phi(g)^{\dagger}\rho(X)\phi(g)\bf{U}(X,\theta)^{\dagger}} \\
&=\trace{\obs(X)\bf{U}(X,\theta)\rho(X)\bf{U}(X,\theta)^{\dagger}} \\
&=\mathcal{L}(X,\boldsymbol{\theta}).
\end{align*}

Case 4:
Assume the observable and state are 
$G$-equivariant, while the circuit is 
$G$-invariant.
Then
\begin{align*}
\mathcal{L}(g \cdot X,\boldsymbol{\theta})
&=\trace{\obs(g \cdot X)\bf{U}( X,\theta)\rho(g \cdot X)\bf{U}( X,\theta)^{\dagger}} \\
&=\trace{\phi(g)\obs(X)\phi(g)^{\dagger}\bf{U}(X,\theta)\phi(g)\rho(X)\phi(g)^{\dagger}\bf{U}(X,\theta)^{\dagger}} \\
&=\trace{\obs(X)\phi(g)^{\dagger}\bf{U}(X,\theta)\phi(g)\rho(X)\phi(g)^{\dagger}\bf{U}(X,\theta)^{\dagger}\phi(g)} \\
&=\trace{\obs(X)\bf{U}(X,\theta)\rho(X)\bf{U}(X,\theta)^{\dagger}} \\
&=\mathcal{L}(X,\boldsymbol{\theta}).
\end{align*}

By cyclicity of the trace, exchanging the roles of $\obs$ and $\rho$ results in identical proofs. Hence the cases where only $\rho$ is equivariant, or where $\pqc$ and $\rho$ are equivariant, follow immediately by symmetry.

The case where all three are invariant is immediate, as is the case where all three are equivariant as the whole product is conjugated by 
$\phi(g)$.

\end{proof}

\begin{proof}[Proof of Proposition \ref{prop:rook_is_symmetric}]

Let $\phi : S_n \to U(\mathcal{H})$ denote the unitary permutation
representation acting by permuting qubits.
By construction, for any $\sigma \in S_n$, the node and edge hamiltonians satisfy
\[
\phi(\sigma)\, H_v \,\phi(\sigma)^\dagger = H_{\sigma(v)}, 
\qquad
\phi(\sigma)\, H_{uv} \,\phi(\sigma)^\dagger = H_{\sigma(u)\sigma(v)} .
\]
Since $\sigma$ relabels the vertex and edge sets according to those of $\sigma \cdot \Gamma$, the sums defining each layer are equivariant under conjugation by $\phi(\sigma)$. Consequently, the circuit obeys the equivariance condition introduced in Section \ref{sec:semi_sym_circs}:
\[
U(\sigma \cdot \Gamma, \boldsymbol{\theta})
=
\phi(\sigma)\, U(\Gamma,\boldsymbol{\theta}) \,\phi(\sigma)^\dagger
\qquad
\forall \, \sigma \in S_n.
\]
Thus, for each fixed data point $\Gamma$, the circuit is $S_n$-equivariant.
\end{proof}

%% file: Paper/appendix/B.tex
\section{B. Implementation Details}
Here we summarise implementation-level details required for full reproducibility of the experimental results.
All experiments were conducted using \texttt{Python} 3.10. Our pipeline is built using \texttt{PyTorch} \cite{pytorch}, with \texttt{TorchQuantum} \cite{torchquantum} as a quantum simulation backend. Graph generation is handled using \texttt{Networkx} \cite{networkx}, and converted into \texttt{TorchGeometric} \cite{torchgeometric} for dataset processing and compatibility with the models.

\subsection{Model Implementations}
Parameters are initialised from a uniform distribution in the range $[0, 0.01]$. The parameters for the initial state are instead sampled from the range $[0, 1]$ in order to span a larger region. 
All quantum circuits are evaluated exactly using statevector methods. When sampling is required we use \texttt{PyTorch} \texttt{multinomial} to sample from a given distribution.

\subsection{Training Procedure}
All models are trained using the Adam optimizer, and learning rate set to $0.01$. We vary the maximum epochs according to the model in order to ensure convergence. The number of epochs and learning rate were selected manually via preliminary testing. We note that learning rates in the range $0.01 - 0.001$ typically perform well.
We use a batch size of $100$, and aggregate the loss and selection accuracies across each batch via the mean.
All runs are seeded using the current timestamp (recorded per run). Experiments are run on the CPU, with multi-threading where implemented by \texttt{TorchQuantum}.

\subsection{Dataset Generation}
Datasets are generated once, and reused across all runs. We sample from the Erd\H{o}s--R\'enyi distribution, with specified values of $n$ and $p$ for each dataset.
The parameters and number of instances per paramter setting for each dataset are summarised in \autoref{tab:model_zoo_kc} and \autoref{tab:model_zoo_mc}.
When sampling graphs, we test for isomorphisms using the Weisfeiler Leman isomorphism test at level 2 \cite{leman_reduction_2018}, and select one from each equivalence class. The labels are computed exactly using built-in \texttt{Networkx} functions.

\subsection{Reproducibility}
All experiments were conducted using fixed random seeds and deterministic training pipelines where possible. 
Full source code and configuration files are available \td{shortly.}

%% file: Paper/appendix/Extra_plots.tex
\section{Learning equivariance}\label{subsec:results-mc-equivariance}
While the \MF models achieved lower training and generalisation performance compared to \Rook models, they did typically surpass random guessing baselines. Here we analyse how well the final models managed to learn the inherent symmetry of the task.

For each model, we construct a permuted version of its training dataset. For each original training instance, we add 2 entries whose nodes and edges have been permuted. Then, we quantify the difference in accuracy that the model displays between the original and permuted datapoints. The results are plotted in \autoref{fig:chantilly_acc}.

\begin{figure}[h]
    \centering
    \includegraphics[width=\linewidth]{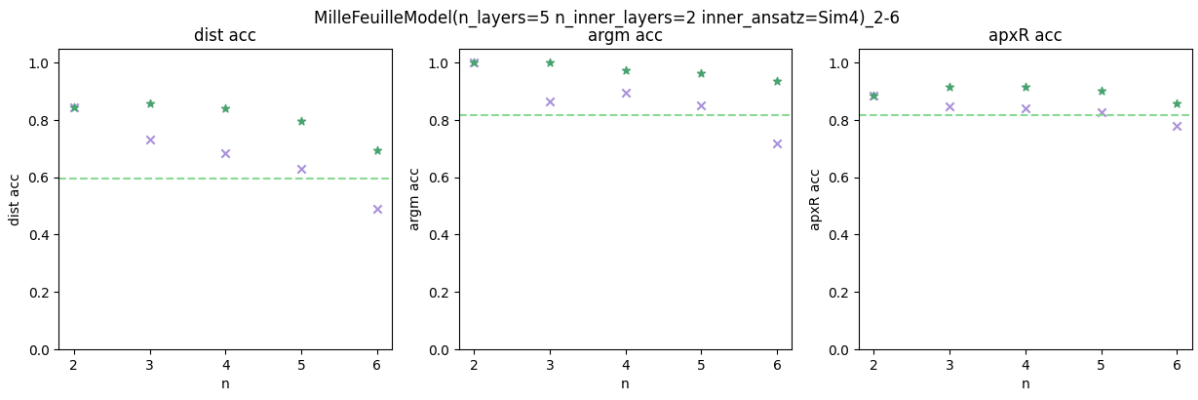}
    \caption{Average \MF model performance drop on a permuted training dataset. Training accuracy is plotted as a star, while the mean permuted accuracy is plotted as a point. Note that for graphs with two nodes, \MF is equivariant as the solutions are symmetric.}
    \label{fig:chantilly_acc}
\end{figure}

Additionally, we inspect directly whether the (anti-)edge hamiltonians at each layer respect the local symmetries pictured in \autoref{fig:gate-symmetry-axioms}. More formally, we expect $H_{u,v} = H_{v,u}$ for undirected graphs as consider here, and that the individual (anti) edges all commute with one another:  $H_{u,v} \cdot H_{u,w} = H_{u,w} \cdot H_{u,v}$. We additionally investigate how close gates are to being the identity. \autoref{fig:gate-symmetry-axioms_plot} visualises the results. We find for both model families that the anti-edge layers are the most non-trivial, and also respect the symmetries the least. Interestingly, the best models do not appear to exhibit better symmetry properties, and we find that the profiles of most models converged to similar points. The \MF models did not converge to parameters that implemented edge unitaries resembling the corresponding trained \Rook edge unitaries.

\begin{figure}[h]
    \centering
    \includegraphics[width=0.3\linewidth]{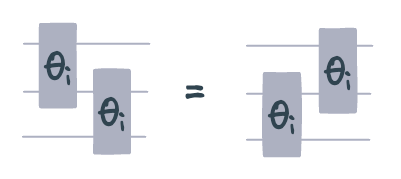}
    \hspace{3em}
    \includegraphics[width=0.3\linewidth]{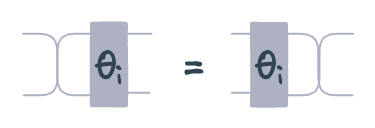}
    \caption{Local edge symmetries. Models that respect these relations will be $S_n$-equivariant and $\mathrm{Aut}(\Gamma)$-invariant.}
    \label{fig:gate-symmetry-axioms}
\end{figure}

\begin{figure}[h]
    \centering
    \includegraphics[width=0.7\linewidth]{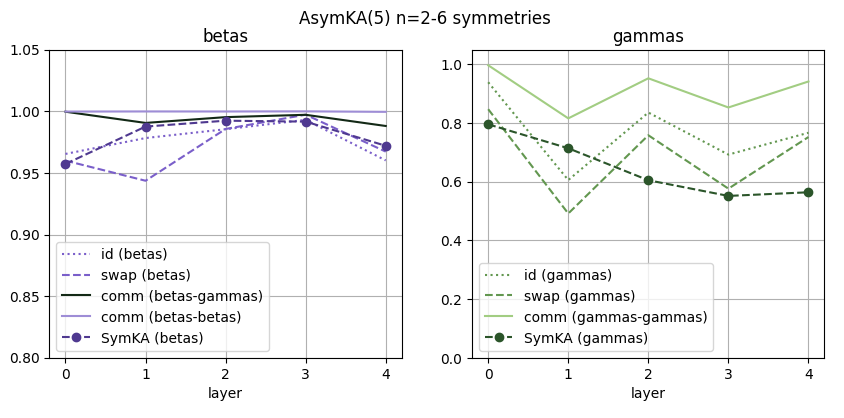}
    \\\vspace{1em}
    \includegraphics[width=0.7\linewidth]{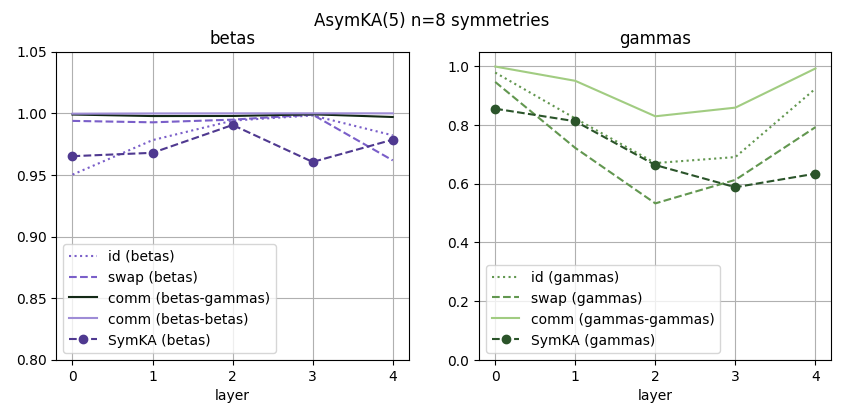}
    \\\vspace{1em}
    \includegraphics[width=\linewidth]{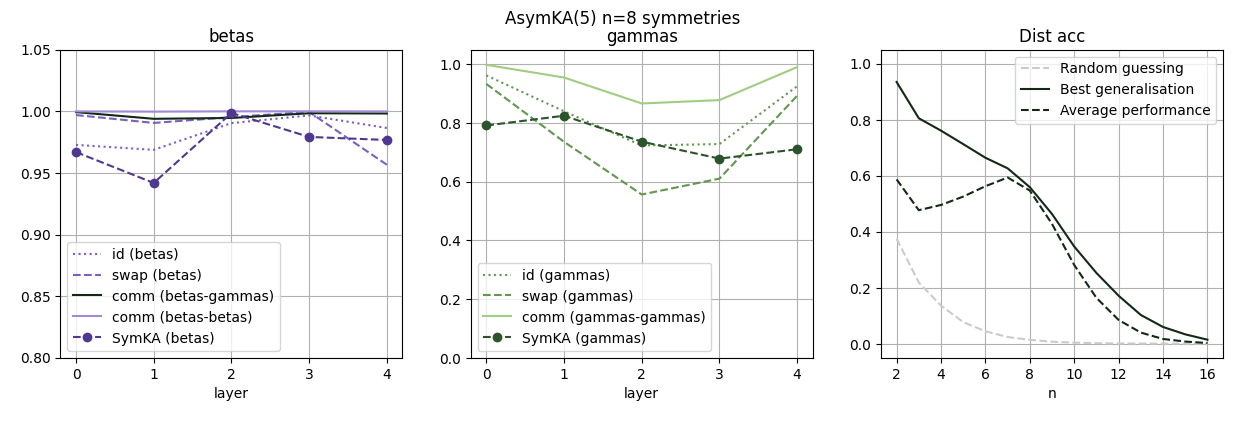}
    
    \caption{Average axiom satisfaction for \MF models. Betas and Gammas correspond to positive and negative edges respectively. In the lower plot, we show one model that performed above average (the dashed line), and its associated gate symmetries.}
    \label{fig:gate-symmetry-axioms_plot}
\end{figure}

\section{Detailed results}
In this section, we display more detailed results for the runs, visualising the individual model performance for each split. We observe that models tend to converge to similar (local) minima, though with different probabilities depending on the choice of dataset and layers.
\begin{figure}[h]
    \centering
    \includegraphics[width=\linewidth]{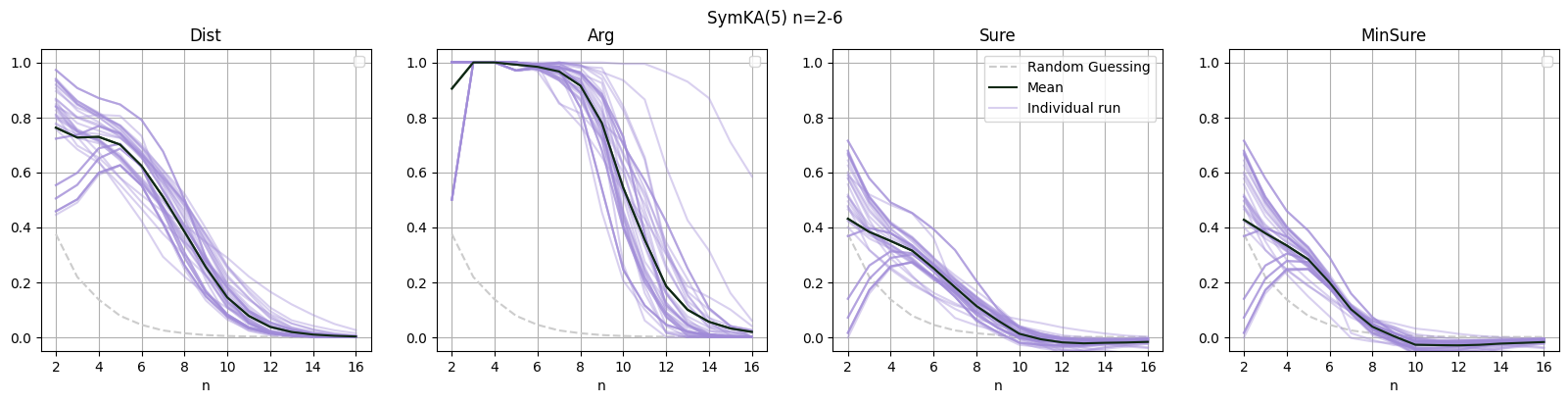}
    \\\vspace{1em}
    \includegraphics[width=\linewidth]{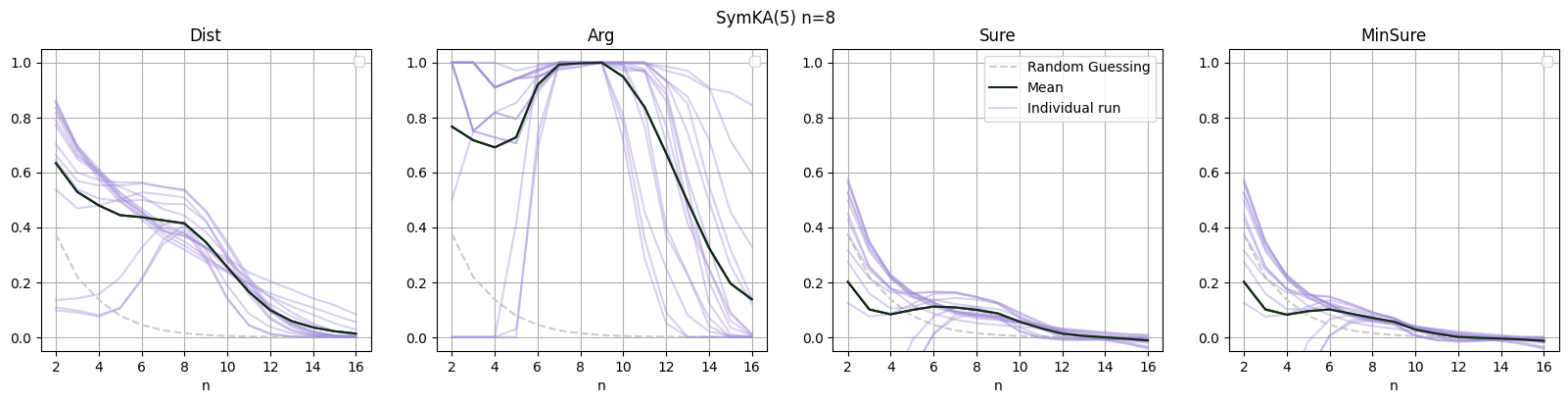}
    \\\vspace{1em}
    \includegraphics[width=\linewidth]{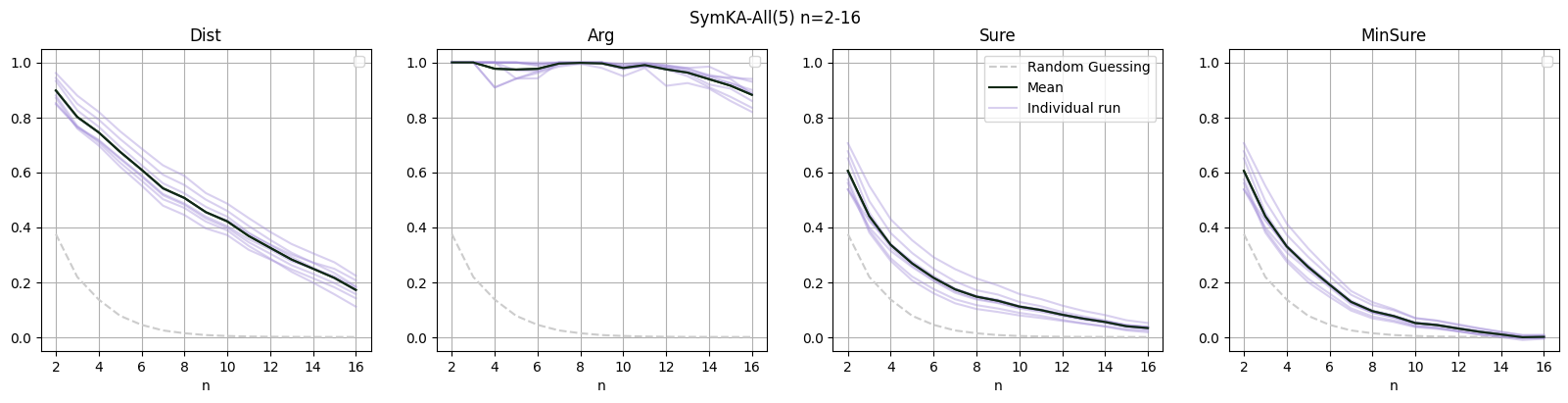}
    \\\vspace{1em}
    \includegraphics[width=\linewidth]{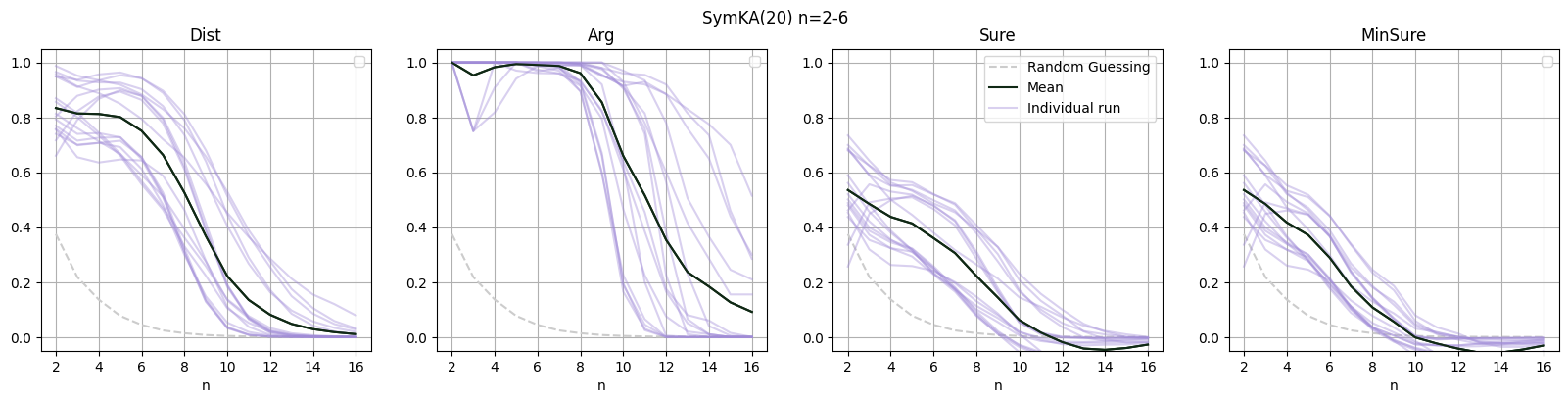}
    \\\vspace{1em}
    \includegraphics[width=\linewidth]{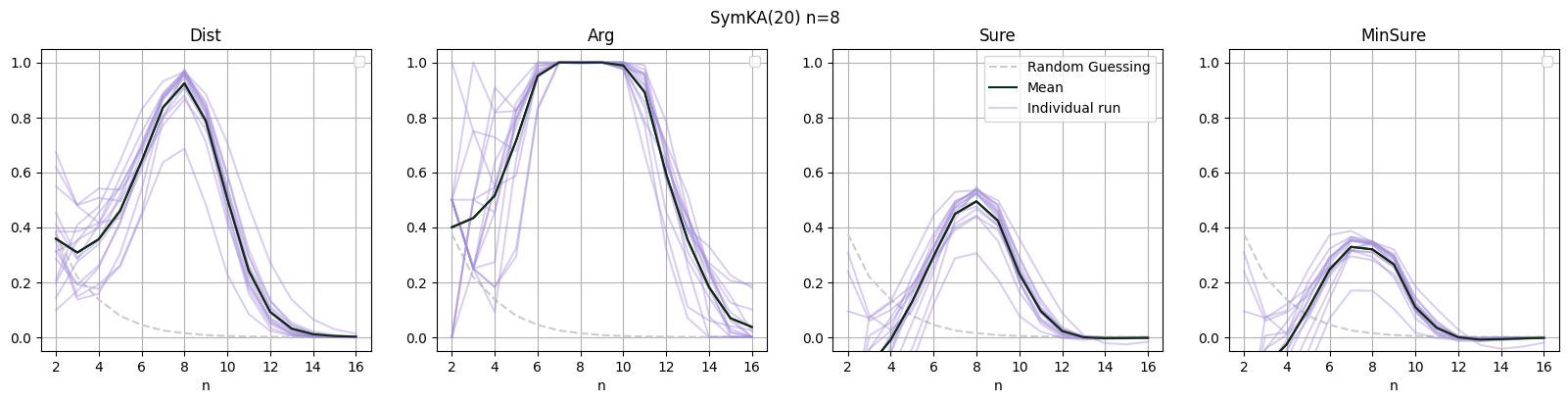}
    \caption{Individual \Rook model results. Each line is a single 5-fold cross-validation iteration for one of the splits. We target at least two iterations per split. The dashed black line shows the mean model performance over the splits.}
    \label{fig:per_model_results:rook}
\end{figure}

\begin{figure}[h]
    \centering
    \includegraphics[width=\linewidth]{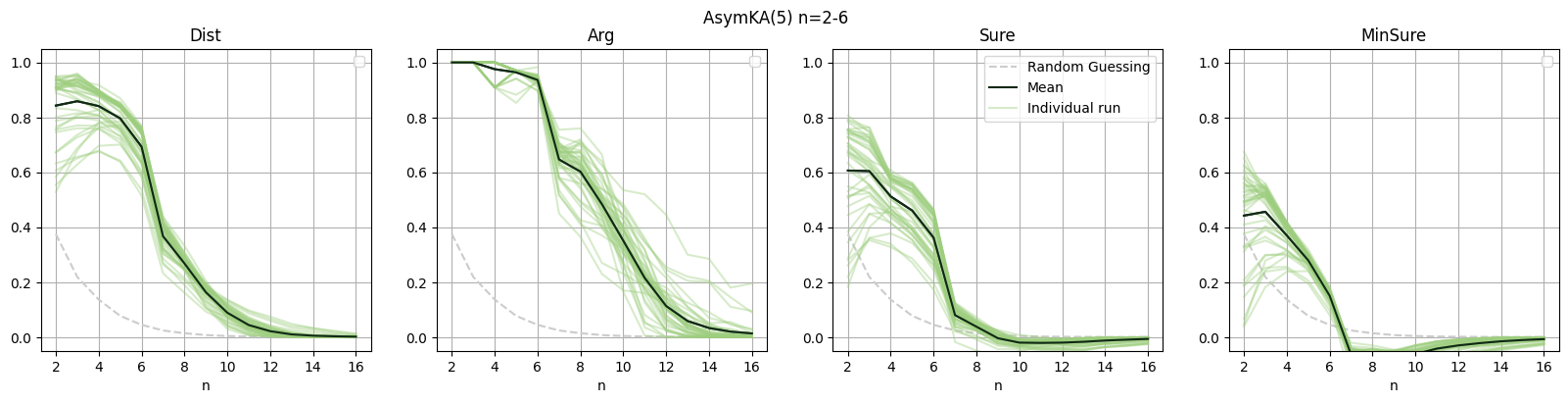}
    \\\vspace{1em}
    \includegraphics[width=\linewidth]{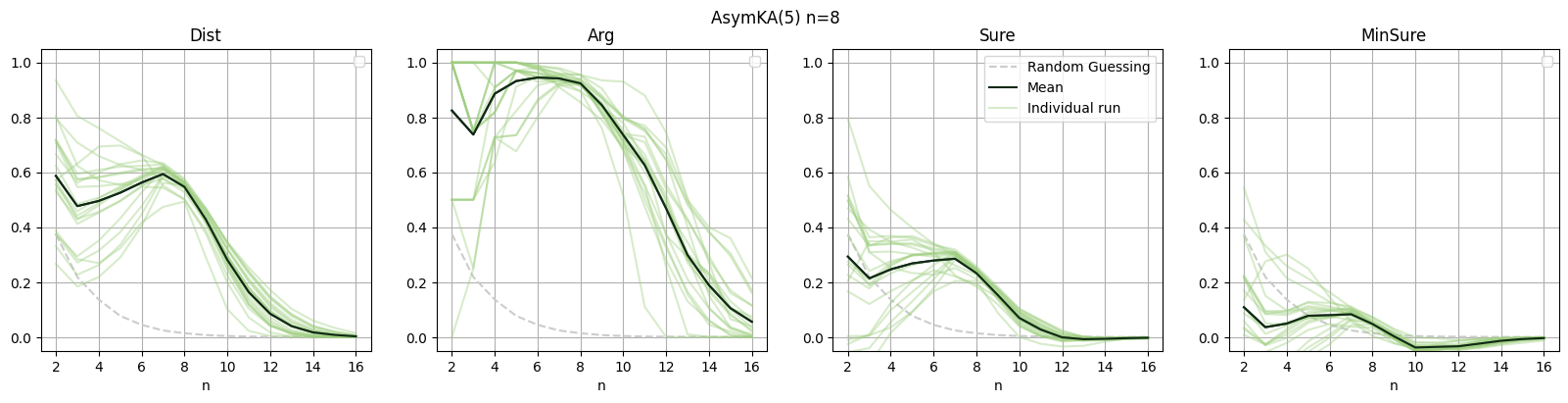}
    \\\vspace{1em}
    \includegraphics[width=\linewidth]{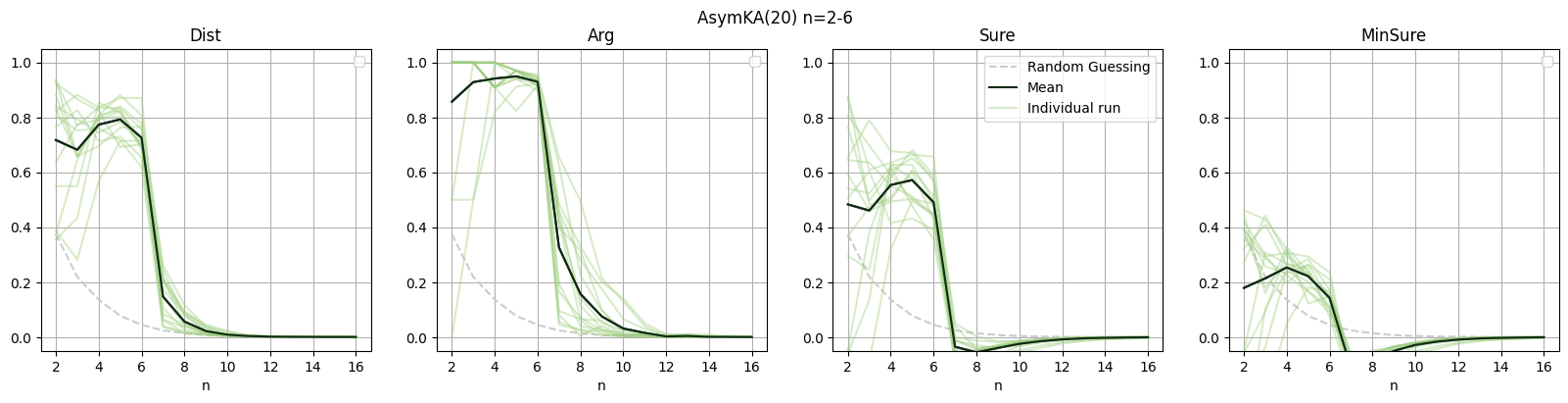}
    \\\vspace{1em}
    \includegraphics[width=\linewidth]{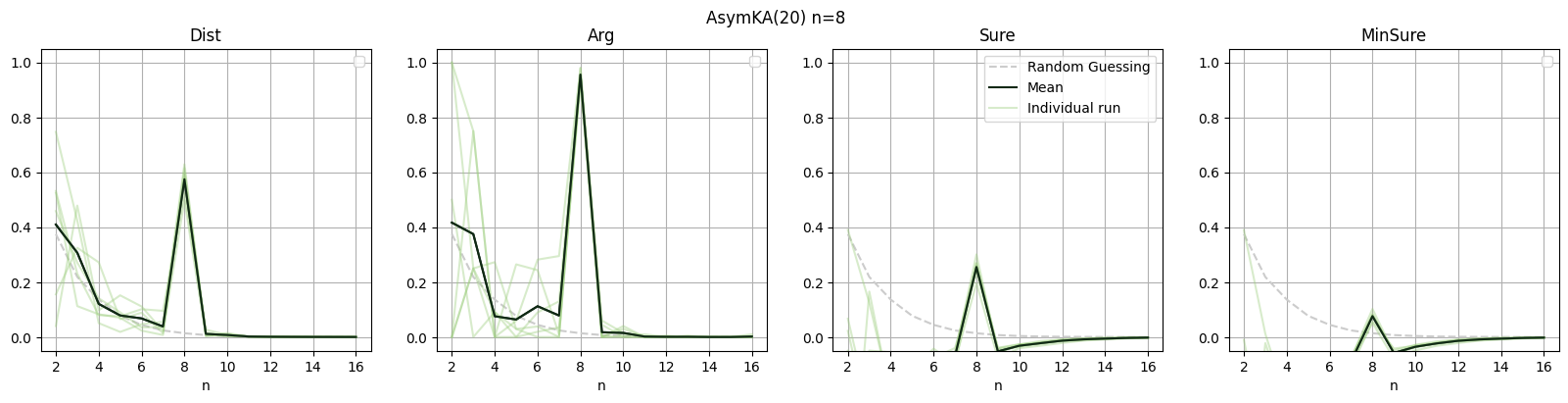}
    \caption{Individual \MF model results. Each line is a single 5-fold cross-validation iteration for one of the splits. We target at least two iterations per split. The dashed black line shows the mean model performance over the splits.}
    \label{fig:per_model_results:mf}
\end{figure}

\begin{figure}[h]
    \centering
    \includegraphics[width=0.8\linewidth]{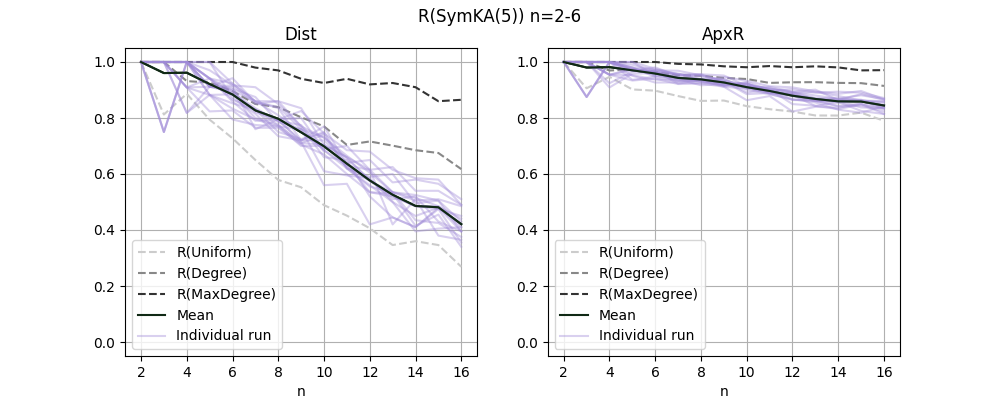}
    \\\vspace{1em}
    \includegraphics[width=0.8\linewidth]{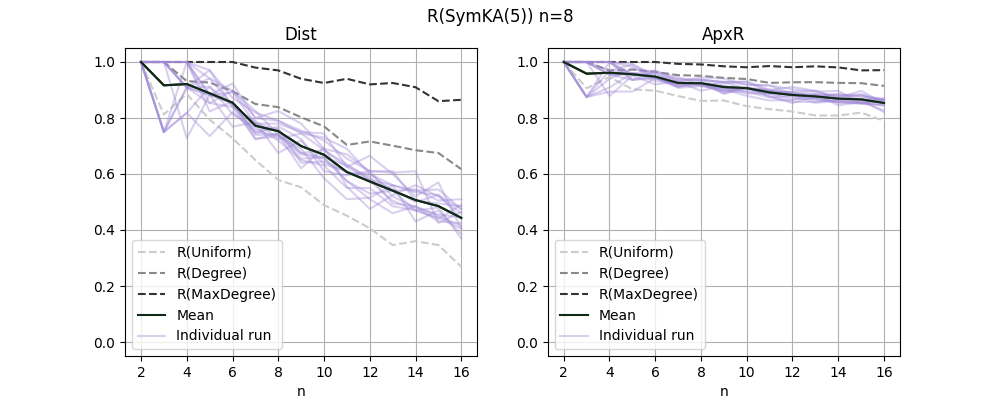}
    \\\vspace{1em}
    \includegraphics[width=0.8\linewidth]{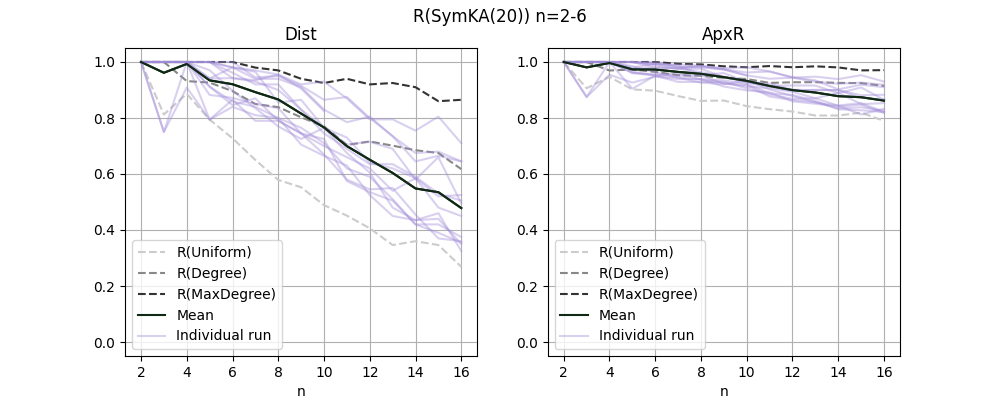}
    \\\vspace{1em}
    \includegraphics[width=0.8\linewidth]{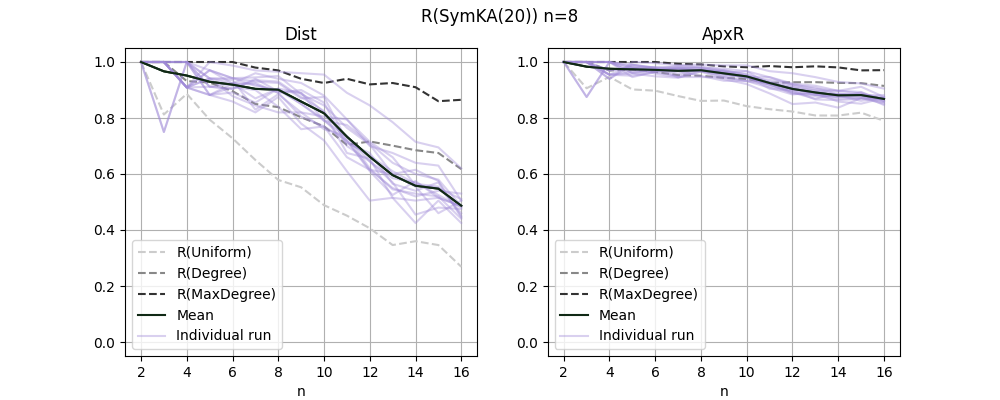}

    \caption{Individual \Pine results for \Rook models. Each line is a single 5-fold cross-validation iteration for one of the splits. We target at least two iterations per split. The dashed black line shows the mean model performance over the splits.}
    \label{fig:per_model_results:pine}
\end{figure}
\begin{figure}[h]
    \centering
    \includegraphics[width=0.8\linewidth]{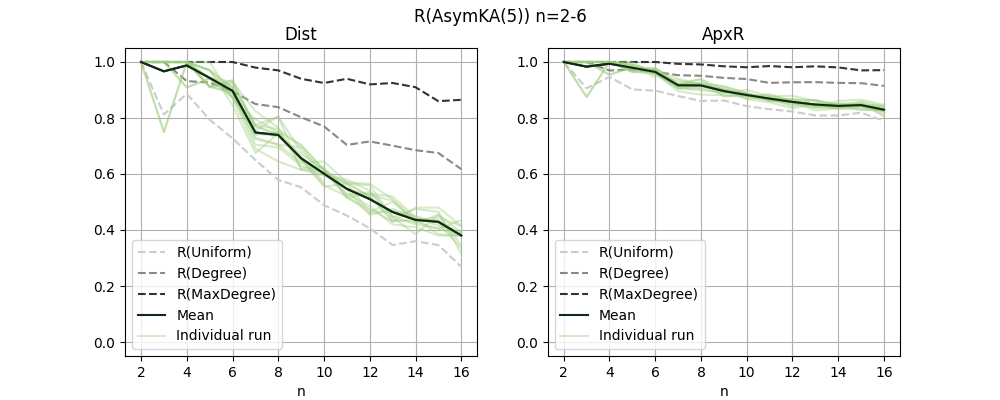}
    \\\vspace{1em}
    \includegraphics[width=0.8\linewidth]{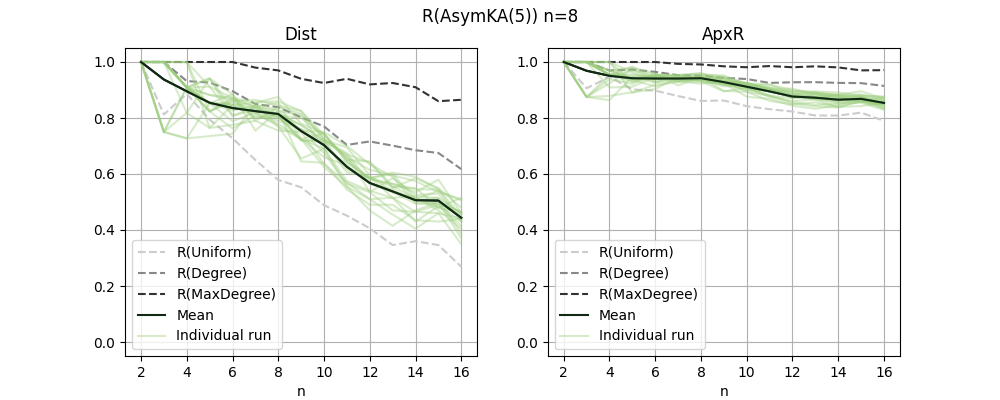}
    
    \caption{Individual \Pine results for \MF models. Each line is a single 5-fold cross-validation iteration for one of the splits. We target at least two iterations per split. The dashed black line shows the mean model performance over the splits.}
    \label{fig:per_model_results:pine}
\end{figure}

%% file: refs.bib
@misc{QAOA,
	title        = {A Quantum Approximate Optimization Algorithm},
	author       = {Edward Farhi and Jeffrey Goldstone and Sam Gutmann},
	year         = 2014,
	url          = {https://arxiv.org/abs/1411.4028},
	eprint       = {1411.4028},
	archiveprefix = {arXiv},
	primaryclass = {quant-ph}
}

@article{yuan2019theory,
	title        = {Theory of variational quantum simulation},
	author       = {Yuan, Xiao and Endo, Suguru and Zhao, Qi and Li, Ying and Benjamin, Simon C.},
	year         = 2019,
	journal      = {Quantum},
	publisher    = {Verein zur F\"{o}rderung des Open Access Publizierens in den Quantenwissenschaften},
	volume       = 3,
	pages        = 191
}

@article{peruzzo2014variational,
	title        = {A variational eigenvalue solver on a photonic quantum processor},
	author       = {Peruzzo, Alberto and McClean, Jarrod and Shadbolt, Peter and Yung, Man-Hong and Zhou, Xiao-Qi and Love, Peter J. and Aspuru-Guzik, Al\'{a}n and O'Brien, Jeremy L.},
	year         = 2014,
	journal      = {Nature Communications},
	publisher    = {Nature Publishing Group},
	volume       = 5,
	number       = 1,
	pages        = 4213
}

@article{dunjko,
	title        = {Quantum machine learning advantages beyond hardness of evaluation},
	author       = {Molteni, Riccardo and Marshall, Simon C. and Dunjko, Vedran},
	year         = 2025,
	journal      = {arXiv preprint arXiv:2504.15964},
	archiveprefix = {arXiv},
	eprint       = {2504.15964}
}

@article{holmes2022connecting,
	title        = {Connecting ansatz expressibility to gradient magnitudes and barren plateaus},
	author       = {Holmes, Zoe and Sharma, Kunal and Cerezo, M. and Coles, Patrick J.},
	year         = 2022,
	journal      = {PRX Quantum},
	publisher    = {Aps},
	volume       = 3,
	number       = 1,
	pages        = {010313}
}

@article{larocca2023theory,
	title        = {Theory of overparameterization in quantum neural networks},
	author       = {Larocca, Martin and Kumar, Nishant and Cerezo, M. and Coles, Patrick J.},
	year         = 2023,
	journal      = {Nature Computational Science},
	publisher    = {Nature Publishing Group},
	volume       = 3,
	pages        = {542--552}
}

@article{marrero2021entanglement,
	title        = {Entanglement-induced barren plateaus},
	author       = {Marrero, Carlos O. and Kieferov{\'a}, M{\'a}ria and Wiebe, Nathan},
	year         = 2021,
	journal      = {PRX Quantum},
	publisher    = {Aps},
	volume       = 2,
	number       = 4,
	pages        = {040316}
}

@inproceedings{takahashi2020classically,
	title        = {Classically Simulating Quantum Circuits with Local Depolarizing Noise},
	author       = {Yasuhiro Takahashi and Yuki Takeuchi and Seiichiro Tani},
	year         = 2020,
	booktitle    = {45th International Symposium on Mathematical Foundations of Computer Science (MFCS 2020)},
	pages        = {83:1--83:13},
	note         = {Also arXiv:2001.08373}
}

@article{jozsa2008matchgates,
	title        = {Matchgates and classical simulation of quantum circuits},
	author       = {Richard Jozsa and Akimasa Miyake},
	year         = 2008,
	journal      = {Proceedings of the Royal Society A},
	volume       = 464,
	number       = 2100,
	pages        = {3089--3106}
}

@article{goh2025lie,
	title        = {Lie-algebraic classical simulations for quantum computing},
	author       = {M. L. Goh and others},
	year         = 2025,
	journal      = {Under review / arXiv},
	note         = {Shows that when the dimension of the algebra grows only polynomially in system size, efficient classical simulation becomes possible}
}

@article{sim_expressibility_2019,
	title        = {Expressibility and entangling capability of parameterized quantum circuits for hybrid quantum-classical algorithms},
	author       = {Sim, Sukin and Johnson, Peter D. and Aspuru-Guzik, Alan},
	year         = 2019,
	month        = dec,
	journal      = {Advanced Quantum Technologies},
	volume       = 2,
	number       = 12,
	pages        = 1900070,
	doi          = {10.1002/qute.201900070},
	issn         = {2511-9044, 2511-9044},
	url          = {http://arxiv.org/abs/1905.10876},
	urldate      = {2023-10-11},
	note         = {arXiv:1905.10876 [quant-ph]}
}

@article{McClean2018BP,
	title        = {Barren plateaus in quantum neural network training landscapes},
	author       = {McClean, Jarrod R. and Boixo, Sergio and Smelyanskiy, Vadim N. and Babbush, Ryan and Neven, Hartmut},
	year         = 2018,
	journal      = {Nature Communications},
	volume       = 9,
	pages        = 4812,
	doi          = {10.1038/s41467-018-07090-4}
}

@article{MBS+18,
	title        = {Barren plateaus in quantum neural network training landscapes},
	author       = {McClean, Jarrod R. and Boixo, Sergio and Smelyanskiy, Vadim N. and Babbush, Ryan and Neven, Hartmut},
	year         = 2018,
	journal      = {Nature Communications},
	volume       = 9,
	number       = 1,
	pages        = 4812,
	doi          = {10.1038/s41467-018-07090-4}
}

@article{Fontana2023,
	title        = {Lie algebraic structure of parameterized quantum circuits and barren plateaus},
	author       = {Fontana, Enrico and Cerezo, M. and Holmes, Zoe and Coles, Patrick J.},
	year         = 2023,
	journal      = {arXiv preprint arXiv:2309.07902}
}

@article{Schatzki2022,
	title        = {Theoretical guarantees for permutation-equivariant quantum neural networks},
	author       = {Schatzki, Louis and Larocca, Martin and Cerezo, M. and Coles, Patrick J.},
	year         = 2022,
	journal      = {arXiv preprint arXiv:2210.09974}
}

@article{QIROA,
	title        = {Quantum-Informed Recursive Optimization Algorithms},
	author       = {Fin\v{z}gar, Jernej Rudi and Kerschbaumer, Aron and Schuetz, Martin J.A. and Mendl, Christian B. and Katzgraber, Helmut G.},
	year         = 2024,
	month        = may,
	journal      = {PRX Quantum},
	publisher    = {American Physical Society (APS)},
	volume       = 5,
	number       = 2,
	doi          = {10.1103/prxquantum.5.020327},
	issn         = {2691-3399},
	url          = {http://dx.doi.org/10.1103/PRXQuantum.5.020327}
}

@article{pullan2006dynamic,
	title        = {Dynamic Local Search for the Maximum Clique Problem},
	author       = {Pullan, Wayne and Hoos, Holger H.},
	year         = 2006,
	journal      = {Journal of Artificial Intelligence Research},
	volume       = 25,
	pages        = {159--185},
	doi          = {10.1613/jair.1815}
}

@article{GENG20075064,
	title        = {A simple simulated annealing algorithm for the maximum clique problem},
	author       = {Xiutang Geng and Jin Xu and Jianhua Xiao and Linqiang Pan},
	year         = 2007,
	journal      = {Information Sciences},
	volume       = 177,
	number       = 22,
	pages        = {5064--5071},
	doi          = {https://doi.org/10.1016/j.ins.2007.06.009},
	issn         = {0020-0255},
	url          = {https://www.sciencedirect.com/science/article/pii/S0020025507002988}
}

@article{evolution,
	title        = {{An Evolutionary Algorithm With Guided Mutation for the Maximum Clique Problem}},
	author       = {{Zhang}, Qingfu and {Sun}, Jianyong and {Tsang}, Edward},
	year         = 2005,
	month        = jan,
	journal      = {IEEE Transactions on Evolutionary Computation},
	volume       = 9,
	number       = 2,
	pages        = {192--200},
	doi          = {10.1109/tevc.2004.840835},
	adsurl       = {https://ui.adsabs.harvard.edu/abs/2005ITEC....9..192Z}
}

@article{adiabatic,
	title        = {Finding cliques by quantum adiabatic evolution},
	author       = {Childs, Andrew M. and Farhi, Edward and Goldstone, Jeffrey and Gutmann, Sam},
	year         = 2002,
	month        = apr,
	journal      = {Quantum Info. Comput.},
	publisher    = {Rinton Press, Incorporated},
	address      = {Paramus, NJ},
	volume       = 2,
	number       = 3,
	pages        = {181–191},
	issn         = {1533-7146},
	issue_date   = {April 2002},
	numpages     = 11
}

@article{dzibuyna,
	title        = {Limitations of tensor-network approaches for optimization and sampling: A comparison to quantum and classical Ising machines},
	author       = {Dziubyna, Anna Maria and Smierzchalski, Tomasz and Gardas, Bartlomiej and Rams, Marek M. and Mohseni, Masoud},
	year         = 2025,
	month        = may,
	journal      = {Physical Review Applied},
	publisher    = {American Physical Society (APS)},
	volume       = 23,
	number       = 5,
	doi          = {10.1103/physrevapplied.23.054049},
	issn         = {2331-7019},
	url          = {http://dx.doi.org/10.1103/PhysRevApplied.23.054049}
}

@article{finzgar_quantum-informed_2024,
	title        = {Quantum-{Informed} {Recursive} {Optimization} {Algorithms}},
	author       = {Fin\v{z}gar, Jernej Rudi},
	year         = 2024,
	journal      = {PRX Quantum},
	volume       = 5,
	number       = 2,
	doi          = {10.1103/PRXQuantum.5.020327}
}

@misc{farhi_quantum_2020,
	title        = {The {Quantum} {Approximate} {Optimization} {Algorithm} {Needs} to {See} the {Whole} {Graph}: {A} {Typical} {Case}},
	shorttitle   = {The {Quantum} {Approximate} {Optimization} {Algorithm} {Needs} to {See} the {Whole} {Graph}},
	author       = {Farhi, Edward and Gamarnik, David and Gutmann, Sam},
	year         = 2020,
	month        = apr,
	publisher    = {arXiv},
	doi          = {10.48550/arXiv.2004.09002},
	url          = {http://arxiv.org/abs/2004.09002},
	urldate      = {2026-03-12},
	note         = {arXiv:2004.09002 [quant-ph]}
}

@inproceedings{esposito_hybrid_2024,
	title        = {Hybrid {Classical}-{Quantum} {Simulation} of {MaxCut} using {QAOA}-in-{QAOA}},
	author       = {Esposito, Aniello and Danzig, Tamuz},
	year         = 2024,
	month        = may,
	booktitle    = {2024 {IEEE} {International} {Parallel} and {Distributed} {Processing} {Symposium} {Workshops} ({IPDPSW})},
	pages        = {1088--1094},
	doi          = {10.1109/ipdpsw63119.2024.00180},
	url          = {http://arxiv.org/abs/2406.17383},
	urldate      = {2026-03-12},
	note         = {arXiv:2406.17383 [quant-ph]},
	language     = {en}
}

@article{bravyi_obstacles_2020,
	title        = {Obstacles to {State} {Preparation} and {Variational} {Optimization} from {Symmetry} {Protection}},
	author       = {Bravyi, Sergey and Kliesch, Alexander and Koenig, Robert and Tang, Eugene},
	year         = 2020,
	month        = dec,
	journal      = {Physical Review Letters},
	volume       = 125,
	number       = 26,
	pages        = 260505,
	doi          = {10.1103/PhysRevLett.125.260505},
	issn         = {0031-9007, 1079-7114},
	url          = {http://arxiv.org/abs/1910.08980},
	urldate      = {2026-01-19},
	note         = {arXiv:1910.08980 [quant-ph]}
}

@article{brady_iterative_2024,
	title        = {Iterative quantum algorithms for maximum independent set},
	author       = {Brady, Lucas T. and Hadfield, Stuart},
	year         = 2024,
	month        = nov,
	journal      = {Physical Review A},
	volume       = 110,
	number       = 5,
	pages        = {052435},
	doi          = {10.1103/PhysRevA.110.052435},
	issn         = {2469-9926, 2469-9934},
	url          = {https://link.aps.org/doi/10.1103/PhysRevA.110.052435},
	urldate      = {2026-02-23},
	language     = {en}
}

@misc{brady_quantum_2025,
	title        = {Quantum {DPLL} and {Generalized} {Constraints} in {Iterative} {Quantum} {Algorithms}},
	author       = {Brady, Lucas T. and Hadfield, Stuart},
	year         = 2025,
	month        = sep,
	publisher    = {arXiv},
	doi          = {10.48550/arXiv.2509.02689},
	url          = {http://arxiv.org/abs/2509.02689},
	urldate      = {2026-03-13},
	note         = {arXiv:2509.02689 [quant-ph]}
}

@misc{wurtz_fixed_2021,
	title        = {The fixed angle conjecture for {QAOA} on regular {MaxCut} graphs},
	author       = {Wurtz, Jonathan and Lykov, Danylo},
	year         = 2021,
	month        = jul,
	publisher    = {arXiv},
	doi          = {10.48550/arXiv.2107.00677},
	url          = {http://arxiv.org/abs/2107.00677},
	urldate      = {2026-03-12},
	note         = {arXiv:2107.00677 [quant-ph]}
}

@inproceedings{networkx,
	title        = {Exploring {Network} {Structure}, {Dynamics}, and {Function} using {NetworkX}},
	author       = {Hagberg, Aric A. and Schult, Daniel A. and Swart, Pieter J.},
	year         = 2008,
	month        = jun,
	address      = {Pasadena, California},
	pages        = {11--15},
	doi          = {10.25080/tcwv9851},
	url          = {https://doi.curvenote.com/10.25080/TCWV9851},
	urldate      = {2026-03-13},
	copyright    = {https://creativecommons.org/licenses/by/3.0/},
	language     = {en}
}

@misc{pytorch,
	title        = {{PyTorch}: {An} {Imperative} {Style}, {High}-{Performance} {Deep} {Learning} {Library}},
	shorttitle   = {{PyTorch}},
	author       = {Paszke, Adam and Gross, Sam and Massa, Francisco and Lerer, Adam and Bradbury, James and Chanan, Gregory and Killeen, Trevor and Lin, Zeming and Gimelshein, Natalia and Antiga, Luca and Desmaison, Alban and K\"{o}pf, Andreas and Yang, Edward and DeVito, Zach and Raison, Martin and Tejani, Alykhan and Chilamkurthy, Sasank and Steiner, Benoit and Fang, Lu and Bai, Junjie and Chintala, Soumith},
	year         = 2019,
	month        = dec,
	publisher    = {arXiv},
	doi          = {10.48550/arXiv.1912.01703},
	url          = {http://arxiv.org/abs/1912.01703},
	urldate      = {2026-03-13},
	note         = {arXiv:1912.01703 [cs]}
}

@misc{torchgeometric,
	title        = {Fast {Graph} {Representation} {Learning} with {PyTorch} {Geometric}},
	author       = {Fey, Matthias and Lenssen, Jan Eric},
	year         = 2019,
	month        = apr,
	publisher    = {arXiv},
	doi          = {10.48550/arXiv.1903.02428},
	url          = {http://arxiv.org/abs/1903.02428},
	urldate      = {2026-03-13},
	note         = {arXiv:1903.02428 [cs]}
}

@inproceedings{torchquantum,
	title        = {{TorchQuantum} {Case} {Study} for {Robust} {Quantum} {Circuits} ({Invited} {Paper})},
	author       = {Wang, Hanrui and Liang, Zhiding and Gu, Jiaqi and Li, Zirui and Ding, Yongshan and Jiang, Weiwen and Shi, Yiyu and Pan, David Z. and Chong, Frederic T. and Han, Song},
	year         = 2022,
	month        = oct,
	booktitle    = {2022 {IEEE}/{ACM} {International} {Conference} {On} {Computer} {Aided} {Design} ({ICCAD})},
	pages        = {1--9},
	url          = {https://ieeexplore.ieee.org/document/10069802},
	urldate      = {2026-03-13},
	note         = {Issn: 1558-2434}
}

@inproceedings{leman_reduction_2018,
	title        = {{The} {Reduction} {Of} {A} {Graph} {To} {Canonical} {Form} {And} {The} {Algebra} {Which} {Appears} {Therein}},
	author       = {Leman, A.},
	year         = 2018,
	url          = {https://www.semanticscholar.org/paper/THE-REDUCTION-OF-A-GRAPH-TO-CANONICAL-FORM-AND-THE-Leman/6d1d91a413af1212fea8791e266282019b62c37d},
	urldate      = {2026-03-13}
}

@article{wybo_missing_2025,
	title        = {Missing {Puzzle} {Pieces} in the {Performance} {Landscape} of the {Quantum} {Approximate} {Optimization} {Algorithm}},
	author       = {Wybo, Elisabeth and Leib, Martin},
	year         = 2025,
	month        = oct,
	journal      = {Quantum},
	volume       = 9,
	pages        = 1892,
	doi          = {10.22331/q-2025-10-22-1892},
	url          = {https://quantum-journal.org/papers/q-2025-10-22-1892/},
	urldate      = {2026-03-25},
	note         = {Publisher: Verein zur F\"{o}rderung des Open Access Publizierens in den Quantenwissenschaften},
	language     = {en-GB}
}

@misc{duneau_scalable_2024,
	title = {Scalable and interpretable quantum natural language processing: an implementation on trapped ions},
	shorttitle = {Scalable and interpretable quantum natural language processing},
	url = {http://arxiv.org/abs/2409.08777},
	doi = {10.48550/arXiv.2409.08777},
	urldate = {2025-06-04},
	publisher = {arXiv},
	author = {Duneau, Tiffany and Bruhn, Saskia and Matos, Gabriel and Laakkonen, Tuomas and Saiti, Katerina and Pearson, Anna and Meichanetzidis, Konstantinos and Coecke, Bob},
	month = sep,
	year = {2024},
	note = {arXiv:2409.08777 [quant-ph]},
}

@misc{recio-armengol_train_2025,
	title = {Train on classical, deploy on quantum: scaling generative quantum machine learning to a thousand qubits},
	shorttitle = {Train on classical, deploy on quantum},
	url = {http://arxiv.org/abs/2503.02934},
	doi = {10.48550/arXiv.2503.02934},
	urldate = {2026-04-20},
	publisher = {arXiv},
	author = {Recio-Armengol, Erik and Ahmed, Shahnawaz and Bowles, Joseph},
	month = mar,
	year = {2025},
	note = {arXiv:2503.02934 [quant-ph]
version: 1},
}

@article{bremner_classical_2011,
	title = {Classical simulation of commuting quantum computations implies collapse of the polynomial hierarchy},
	volume = {467},
	issn = {1364-5021, 1471-2946},
	url = {http://arxiv.org/abs/1005.1407},
	doi = {10.1098/rspa.2010.0301},
	number = {2126},
	urldate = {2025-08-04},
	journal = {Proceedings of the Royal Society A: Mathematical, Physical and Engineering Sciences},
	author = {Bremner, Michael J. and Jozsa, Richard and Shepherd, Dan J.},
	month = feb,
	year = {2011},
	note = {arXiv:1005.1407 [quant-ph]},
	pages = {459--472},
}

@misc{placidi_impact_2026,
	title = {The {Impact} of {Qubit} {Connectivity} on {Quantum} {Advantage} in {Noisy} {IQP} {Circuits}},
	url = {http://arxiv.org/abs/2604.12635},
	doi = {10.48550/arXiv.2604.12635},
	urldate = {2026-04-20},
	publisher = {arXiv},
	author = {Placidi, Leonardo and Rinaldi, Enrico and Fujii, Keisuke and Liu, Chen-Yu},
	month = apr,
	year = {2026},
	note = {arXiv:2604.12635 [quant-ph]},
}
